\documentclass[onecolumn,11pt,aps,nofootinbib,superscriptaddress,tightenlines]{revtex4}

\usepackage{graphics}
\usepackage{dsfont}
\usepackage{amsmath}
\usepackage{amssymb}
\usepackage{amsthm}
\usepackage{graphicx}
\usepackage{mathrsfs}
\usepackage{units}
\usepackage[usenames,dvipsnames]{color}
\newcommand{\abs}[1]{\ensuremath{|#1|}}

\newcommand{\Norm}[2]{\ensuremath{\left|\!\left|#1\right|\!\right|_{#2}}}

\newcommand{\tr}{\textnormal{tr}}
\newcommand{\trace}[1]{\ensuremath{\tr (#1)}}

\newcommand{\ket}[1]{| #1 \rangle}

\newcommand{\bra}[1]{\langle #1 |}

\newcommand{\braket}[2]{\langle #1 | #2 \rangle}

\newcommand{\bracket}[3]{\langle #1 | #2 | #3 \rangle}

\newcommand{\proj}[2]{| #1 \rangle\!\langle #2 |}

\newcommand{\id}{\ensuremath{\mathds{1}}}

\newcommand{\opid}{\ensuremath{\mathcal{I}}}

\newcommand{\cC}{\mathcal{C}}

\newcommand{\cE}{\mathcal{E}}
\newcommand{\cF}{\mathcal{F}}
\newcommand{\cG}{\mathcal{G}}

\newcommand{\cM}{\mathcal{M}}

\newcommand{\cS}{\mathcal{S}}
\newcommand{\cT}{\mathcal{T}}
\newcommand{\cU}{\mathcal{U}}

\theoremstyle{plain}
\newtheorem{lemma}{Lemma}
\newtheorem{theorem}[lemma]{Theorem}

\theoremstyle{definition}
\newtheorem{definition}{Definition}

\begin{document}
\title{Connected components of irreducible maps and 1D quantum phases}
\author{Oleg Szehr}
\email{oleg.szehr@posteo.de}
\affiliation{Centre for Quantum Information, University of Cambridge,
Cambridge CB3 0WA,
United Kingdom}
\author{Michael M. Wolf}
\email{wolf@ma.tum.de}
\affiliation{Zentrum Mathematik, Technische Universit\"{a}t M\"{u}nchen, 85748 Garching, Germany}
\date{\today}
\begin{abstract}

We investigate elementary topological properties of sets of completely positive (CP) maps that arise in quantum Perron-Frobenius theory. We prove that the set of primitive CP maps of fixed Kraus rank is path-connected and we provide a complete classification of the connected components of irreducible CP maps at given Kraus rank and fixed peripheral spectrum in terms of a multiplicity index. 

These findings are then applied to analyse 1D quantum phases by studying equivalence classes of translational invariant Matrix Product States that correspond to the connected components of the respective CP maps. Our results extend the previously obtained picture in that they do not require blocking of physical sites, they lead to analytic paths and they allow to decompose into ergodic components and to study the breaking of translational symmetry. 
\end{abstract}

\maketitle

\tableofcontents
\section{Introduction}

This work is devoted to the study of elementary topological properties of sets of irreducible completely positive (CP) maps. We analyse in particular the connected components of the sets of irreducible maps with given Kraus rank and fixed peripheral spectrum. The motivation for this analysis is the intention to contribute to the understanding and characterization of  phase diagrams of quantum many-body systems in 1D. 

The link between connected components of classes of CP maps and 1D quantum phases is given by the framework of Matrix Product States (MPS~\cite{MPS}) and Finitely Correlated States (FCS~\cite{MPSwerner}). On the one hand, these allow non-trivial statements about 1D quantum phases via a frustration-free parent Hamiltonian model \cite{MPS,MPSwerner,Nacht1} and constitute a natural framework for the analysis of gapped 1D spin systems~\cite{Verstr12,Has1,Has2}. On the other hand, they are characterized by CP maps with irreducible maps being the essential elementary building blocks.

Commonly, the ground states of two gapped 1D system are said to belong to the same \lq\lq{}quantum phase\rq\rq{} iff they can be connected by a path of gapped local Hamiltonians. That is, the ground states of ~$H^{(0)}=\sum_i  \tau^i(h^{(0)})$ and $H^{(1)}=\sum_i \tau^i(h^{(1)})$ belong to the same quantum phase iff
there is a family of local interactions $h^{(t)}$ that is continuous in $t\in[0,1]$ so that the global Hamiltonian $H^{(t)}=\sum_i  \tau^i(h^{(t)})$ has a non-vanishing energy gap above the set of ground states~\cite{phasesNi,phases,phasesBa,Bach, Bach1, Bach2}. Here $\tau^i$ denotes a lattice translation. Along such a path, expectation values of local observables change continuously while the structure of the ground state subset remains unchanged. 

Depending on whether or not one considers blocking of sites, symmetry protection, stability, extension of the local Hilbert spaces or focusses on order parameters or ground state manifolds, etc. there are arguably several meaningful ways to define a quantum phase, which in addition have to cope with the general undecidability of spectral gaps \cite{undecidablegap} and phase transitions within the MPS framework~\cite{WOVC}.

For this reason we prefer to abstain from defining the notion of a quantum phase, but speak instead of \emph{equivalence classes} of local gapped Hamiltonians under
\emph{analytic} deformations of local interactions. Roughly speaking, we consider $h^{(0)}$ and $h^{(1)}$ to be equivalent iff an analytic interpolating Hamiltonian path $h^{(t)}$ exists so that $H^{(t)}=\sum_i \tau^i(h^{(t)})$ is gapped even in the limit of an infinite chain (cf. Sec.\ref{sec:phases}). To provide sufficient conditions for equivalence and to construct the respective paths $h^{(t)}$ we then rely on our insights about connectedness of sets of CP maps. 

A first classification of 1D quantum phases within the MPS framework was given in~\cite{phasesNi,phases}. Our work provides a more detailed picture of this classification by providing an exhaustive analysis on the level of CP maps. Translated to the MPS framework this implies that, in contrast to~\cite{phasesNi,phases}, we do not need to \emph{block} physical sites nor do we change the physical dimension of sites along the interpolating path. In this way, our construction automatically yields Hamiltonian paths that preserve translational symmetry as well as the breaking of this symmetry within the space of ground states. 

A very similar approach to ours was recently taken in the independent contribution~\cite{phasesBa}.
There the authors also prove path-connectedness of primitive CP maps under fixed Kraus rank, albeit only with $\cC^{1}$-paths, and apply their finding to the classification of 1D quantum phases. Showing the result via a different technique enables us to give a shorter proof that immediately implies the existence of analytic paths, allows to fix the Kraus rank (rather than to bound it from above) and yields an extension to more general irreducible maps, which were not studied in~\cite{phasesBa}.

\section{Topological Properties of sets of Completely Positive Maps}
\subsection{Preliminaries on Completely Positive Maps}
\label{prel:not}

We study linear maps $\cT:\cM_{D}(\mathbb{C})\rightarrow\cM_{D}(\mathbb{C})$ acting on the set of complex $D\times D$ matrices $\cM_{D}(\mathbb{C})$. The spectrum $\rho(\cT)$ of $\cT$ is the set of $\lambda\in\mathbb{C}$ so that $\lambda\opid-\cT$ is not invertible, where $\opid:\cM_{D}(\mathbb{C})\rightarrow\cM_{D}(\mathbb{C})$ is the identity map. The spectral radius is defined by
$\mu=\max\{\abs{\lambda}\:|\:\lambda\in\rho(\cT)\}$ and the peripheral spectrum is the set of eigenvalues of magnitude $\mu$. Maps of the structure $\cT(X):=\sum_i A_iXA_i^\dagger$ with $\{A_i\}_i\in\cM_{D}(\mathbb{C})$ are called completely positive (CP) and the operators $\{A_i\}_i$ are referred to as Kraus operators. We denote the set of completely positive maps on $\cM_{D}(\mathbb{C})$ by $\mathfrak{T}$. The minimal number of Kraus operators required to represent $\cT$ is the Kraus rank of $\cT$ and we write $\mathfrak{T}_{(r)}$ for the set of CP maps of fixed Kraus rank $r$. For each CP map the dual map $\cT^*$ is given by $\cT^*(X):=\sum_i A_i^\dagger XA_i$. $\cT^*$ is simply the adjoint of $\cT$ with respect to the Hilbert-Schmidt inner product $(X,Y)\mapsto\trace{X^{\dagger} Y}$. A completely positive map $\cT$ is called unital (CPU) iff it preserves the identity operator $\cT(\id)=\id$ and trace-preserving (CPTP) iff $\cT^*(\id)=\id$. By Gelfand's formula CPU (and CPTP) maps have spectral radius $1$. For any linear map $\cT$ the Choi matrix representation is $\omega(\cT):=(\cT\otimes\opid)(\proj{\varphi}{\varphi})$ with $\varphi=\sum_i\ket{ii}$. It is well known that $\cT$ is completely positive iff $\omega(\cT)$ is positive semi-definite and that the Kraus rank of $\cT$ equals $\textnormal{rank}(\omega(\cT))$. In the following we introduce some core concepts from the so-called Perron-Frobenius theory of positive maps. Later, these concepts will be pivotal to our study of 1D quantum phases. See~\cite{specevans,farenick,albev,Michael:Skript} for an introduction to the topic.
\begin{definition}[Irreducibility~\cite{Daviesirred,specevans,albev}]\label{irreducible}
A completely positive map $\cT$ on $\cM_D(\mathbb{C})$ is called \emph{irreducible} if one of the following equivalent properties holds.
\begin{enumerate}
\item $\cT$ has no non-trivial reducing hereditary subalgebras of $\cM_D(\mathbb{C})$. That is, if $P\in\cM_D{(\mathbb{C})}$ is a Hermitian projector with $\cT(P\cM_D{(\mathbb{C})}P)\subseteq P\cM_D{(\mathbb{C})}P$ then $P\in\{0,\id\}$.
\item The spectral radius of $\cT$ is a non-degenerate eigenvalue and the corresponding left- and right- eigenvectors are positive definite matrices.
\end{enumerate}
\end{definition}
We will denote by $\mathfrak{I}\subseteq\mathfrak{T}$ the set of irreducible CP maps on $\cM_{D}(\mathbb{C})$ and by $\mathfrak{I}_{(r)}\subseteq\mathfrak{T}_{(r)}$ the respective subset of maps of Kraus rank $r$.
Spectral properties of positive maps on $C^*$-algebras have been studied in detail in~\cite{specevans} and we will heavily draw on the methods of this article. See~\cite[Theorem~2.3-2.4]{specevans} and~\cite{Michael:Skript} for the equivalence of \emph{1.} and \emph{2.}. Of particular importance will be the class of irreducible CP maps with trivial peripheral spectrum, where~every eigenvalue that is not the spectral radius has strictly smaller magnitude:
\begin{definition}[Primitivity]\label{primitive}
A completely positive map $\cT$ on $\cM_D(\mathbb{C})$ is called \emph{primitive} if one of the following equivalent properties holds.
\begin{enumerate}
\item There is $n\in\mathbb{N}$ so that $\textnormal{rank}(\omega(\cT^n))=D^2$.
\item We have $\textnormal{rank}(\omega(\cT^{D^4}))=D^2$.
\item The spectral radius of $\cT$ is a non-degenerate eigenvalue, all other eigenvalues have strictly smaller magnitude, and the left- and right- eigenvectors corresponding to the spectral radius are positive definite matrices.
\end{enumerate}
\end{definition}
We write $\mathfrak{P}\subseteq\mathfrak{I}$ for the set of primitive CP maps and $\mathfrak{P}_{(r)}\subseteq\mathfrak{I}_{(r)}$ stands for the subset of maps of Kraus rank $r$.
Point \emph{1.}~in Definition~\ref{primitive}~says that a certain power of $\cT$ has maximal Kraus rank. This is equivalent to saying that the list of matrix products $\{A_{i_1}\cdot A_{i_2}\cdot...\cdot A_{i_n}\}_{i_k=1,...,r}$, where $r$ is the Kraus rank of $\cT$, spans the entire algebra $\cM_D(\mathbb{C})$. Note that if \emph{1.} holds for a given $n_0$ then it does so for any $n\geq n_0$. If $\cT$ is CPTP and primitive, point \emph{3.}~implies that for any initial state the semi-group $\{\cT^n\}_{n\geq0}$ converges to a unique full-rank stationary state~\cite{szrewo}. 

For CPTP maps the equivalence of \emph{1.-3.}~is established in~\cite{Wielandt}. In particular, the implication from \emph{2.}~to \emph{1.}~is a consequence of quantum Wielandt's inequality. Our definition goes slightly beyond the framework of~\cite{Wielandt} as we drop the TP requirement. However, a marginal modification of the proofs establishes the equivalence \emph{1.-3.} (cf.~Appendix~\ref{equivalpr} for completeness).

\begin{lemma}[\cite{specevans}]\label{irredprop}
Let $\cT:\cM_D{(\mathbb{C})}\rightarrow \cM_D{(\mathbb{C})}$ with $\cT(X)=\sum_i A_i^\dagger XA_i$ be irreducible and unital. It follows that:
\begin{enumerate}
\item The peripheral spectrum of $\cT$ is $\{e^{2\pi i \frac{k}{m}}\}_{k=1,...,m}$ for some $m\in\{1,...,D\}$ and each of these eigenvalues has algebraic multiplicity $1$.
\item There is a unitary $U$ with $\cT(U^k)=\beta^k U^k\ \forall k\in\mathbb{Z}$, where $\beta:=e^{2\pi i \frac{k}{m}}$.
\item $U$ has spectral decomposition $U=\sum_{k=1}^{m}\beta^k P_k$ with spectral projectors $P_k$ of dimension $d_k=\trace{P_k}$ and $\sum_{k=1}^{m} d_k=D$.
\item For the Kraus operators it holds that $A_i P_{k-1}=P_{k}A_i$ for all $i$ and $k$, where $P_0:=P_m$.
\end{enumerate}
\end{lemma}
We call a CP map $\cT$ \emph{irreducible of degree $m$} if it is irreducible and the peripheral spectrum contains $m$ eigenvalues and we write $\mathfrak{I}_{(r,m)}\subseteq\mathfrak{I}_{(r)}$ for the set of irreducible maps of Kraus rank $r$ and degree $m$. Clearly we have $\mathfrak{I}_{(r,1)}=\mathfrak{P}_{(r)}$. We observe that $\cT$ is irreducible of degree $m$ iff $\cT^*$ has this property. The intertwining relation \emph{4.}~has the important consequence that it reveals the structure of Kraus operators of irreducible CPU maps. All Kraus operators can simultaneously be written in the form
\begin{align*}
A_i=
\begin{pmatrix}
0 & 0 & \cdots & 0 & A_i^{(m)}\\
A_i^{(1)} & 0 & \cdots & 0 & 0\\
0 & A_i^{(2)} & \ddots & \vdots & \vdots\\
\vdots & \ddots & \ddots & 0 & 0\\
0 & \cdots & 0 & A_i^{(m-1)} & 0\\
\end{pmatrix}
\end{align*}
with $d_{k+1}\times d_{k}$ submatrices $A_i^{(k)}$ (with respect to a basis that diagonalizes $U$).
This standard form is referred to as the \emph{Frobenius form}. Finally, note the following lemma, which provides a relation between irreducible and primitive CP maps.
\begin{lemma}\label{mini}\

\begin{enumerate}
\item Let $\bar{\cT}=\frac{1}{2}(\cT+\cT^2)$. We have $\cT\in\mathfrak{I}$ iff
$\bar{\cT}\in\mathfrak{P}$.
\item Let $\cT\in\mathfrak{I}$. We have that $\cT$ is irreducible of degree $m$ iff there is a set of exactly $m$ (not necessary Hermitian) projectors $\{\tilde{P}_k\}_{k=1,...,m}$ with 
$\sum_{k=1}^m \tilde{P}_k=\id$ and $\tilde{P}_k\neq0$ so that $\cT^m(\tilde{P}_k\cM_D{(\mathbb{C})}\tilde{P}_k)\subseteq \tilde{P}_k\cM_D{(\mathbb{C})}\tilde{P}_k$ and $\cT^m |_{\tilde{P}_k\cM_D{(\mathbb{C})}\tilde{P}_k}$ is primitive.
\end{enumerate}
\end{lemma}
We prove Lemma~\ref{mini} in Appendix~\ref{AppB}.

\subsection{Irreducible maps with fixed Kraus rank}

%We study topological properties of sets of completely positive maps, which will be relevant for our classification of quantum phases in MPS systems. 

The set $\mathfrak{T}$ of completely positive maps on $\cM_D(\mathbb{C})$ constitutes a closed convex cone. The Kraus rank of a CP map $\cT$ equals the rank of the corresponding Choi matrix $\omega(\cT)$ and is bounded by $D^2$. It follows that the boundary of $\mathfrak{T}$ consists of CP maps with Kraus rank $r<D^2$. Furthermore the boundary can be subdivided into convex faces corresponding to rank deficient $\omega$ with identical kernel. Any map in the interior of $\mathfrak{T}$ has maximal Kraus rank and consequently is primitive. Thus $\mathfrak{P}$ constitutes a dense subset of $\mathfrak{T}$ and all non-primitive maps are located at the boundary. In particular, the boundary contains all reducible and all non-primitive irreducible CP maps. A core point in this section will be to study the connectivity of the mentioned subsets of the boundary in both the general and the trace-preserving case. %It will turn out that they constitute a set of measure zero with respect to the primitive maps.
We begin by showing that for any $r\geq2$ there are primitive CPTPU (CPTP and unital) maps in $\mathfrak{T}$ that have Kraus rank $r$. In particular, for any $D\geq2$ the boundary of $\mathfrak{T}$ also contains primitive maps.
\begin{lemma}
\label{non-empty}
If $D,r\geq2$ then the set $\mathfrak{P}_{(r)}$ of primitive CP maps of Kraus rank $r$ on $\cM_{D}(\mathbb{C})$ is not empty.
\end{lemma}
\begin{proof}
We construct the respective channels relying on the so-called discrete Weyl system. We define $U_0,U_1\in\cM_D(\mathbb{C})$ by
\begin{align*}
U_0=\sum_{k=0}^{D-1}e^{\frac{2\pi k i}{D}}\proj{k}{k}\quad\textnormal{and}\quad U_1=\sum_{k=0}^{D-1}\proj{k+1}{k}.
\end{align*}
It is well known~\cite{wernerall, Michael:Skript} that the set of $D^2$ matrix products $\{U_0^iU_1^j\}_{i,j=0,...,D-1}$ constitutes an orthogonal basis for $\cM_D(\mathbb{C})$. For any $3\leq r\leq D^2$ we construct a CPTPU map $\cT\in\mathfrak{P}_{(r)}$ by choosing a set of $r$ Kraus operators of the form $\frac{1}{\sqrt{r}}U_0^iU_1^j$ that contains $\{\frac{1}{\sqrt{r}}\id,\frac{1}{\sqrt{r}}U_0,\frac{1}{\sqrt{r}}U_1\}$. In other words
$$\cT(X)=\frac{1}{r}\sum_{(i,j)\in S}U_0^iU_1^jX(U_0^iU_1^j)^\dagger,$$
where $S\subseteq\{0,...,D-1\}^2$ is an index set with $r$ elements and $\{(0,0),(0,1),(1,0)\}\subseteq S$. To see that $\cT$ is primitive consider $r=3$. In this case
the list of $3^{2D-2}$ matrix products $\{A_{i_1}A_{i_2}... A_{i_{2D-2}}\}$ formed from $A_i\in\{\id,U_0,U_1\}$ spans $\cM_D(\mathbb{C})$ as it contains $\{U_0^iU_1^j\}_{i,j=0,...,D-1}$ as a subset. For $r>3$ choosing Kraus operators of the form $\frac{1}{\sqrt{r}}U_0^iU_1^j$ guarantees that $\cT$ has Kraus rank $r$, while the set of matrix products formed from all Kraus operators still spans the entire matrix algebra. For $r=2$ the situation is slightly more complicated. In this case we choose $U_1$ as above but replace $U_0$ by
$$\hat{U}_0=\sum_{k=0}^{D-1}e^{(k+1)i}\proj{k}{k}.$$
Note that the diagonal entries of $\hat{U}_0$ correspond to complex rotations by \emph{irrational} angles. As a consequence for any nonzero $m\in\mathbb{N}$ we have $\hat{U}_0^m\neq\id$, while ${U}_0^D={U}_1^D=\id$. We consider a list of matrix products $\{A_{i_1}A_{i_2}... A_{i_{D^2}}\}$, where $A_i\in\{\hat{U}_{0},{U}_1\}$ and observe that it contains
$$\{\hat{U}_0^{D^2-mD-n}{U}_1^{mD+n}\}_{m,n=0,...,D-1}=\bigcup_{n=0}^{D-1}\left(\{\hat{U}_0^{D^2-mD-n}{U}_1^{n}\}_{m=0,...,D-1}\right)$$
as a subset. The latter set already spans the whole matrix algebra $\cM_D(\mathbb{C})$. To see this we note that the matrix ${U}_1^{n}$ has $D^2-D$ entries that are $0$ and $D$ entries that are $1$. Furthermore, as can be verified easily by direct computation, if for a given $n\in\{0,...,D-1\}$ the matrix ${U}_1^{n}$ has an entry $1$ at a certain position then for any $n'\in\{0,...,D-1\}$, $n'\neq n$ the matrix ${U}_1^{n'}$ has an entry $0$ at this position, i.e.~$\forall k,l,n,n'\in\{0,...,D-1\}$, $n\neq n':$ $\bracket{k}{U_1^n}{l}=1\Rightarrow\bracket{k}{U_1^{n'}}{l}=0$. Consequently the $D$ matrices of the form $U_1^n$ are linearly independent. Since $\hat{U}_0^{p}{U}_1^{n}$, $p\in\{0,...,D-1\}$ has non-zero entries at the same positions as $U_1^{n}$ it follows that for $n\neq n'$ the spaces spanned by $\{\hat{U}_0^{D^2-mD-n}{U}_1^{n}\}_{m=0,...,D-1}$ and $\{\hat{U}_0^{D^2-mD-n'}{U}_1^{n'}\}_{m=0,...,D-1}$ are linearly independent.
%
%
%
%$$\textnormal{span}_{m,n=0,...,D-1}\{\hat{U}_0^{D^2-mD-n}{U}_1^{mD+n}\}=\bigoplus_{n=0}^{D-1}\textnormal{span}_{m=0,...,D-1}\{\hat{U}_0^{D^2-mD-n}{U}_1^{n}\}$$
%
%
%
It remains to verify that for fixed $n$ the products $\hat{U}_0^{D^2-mD-n}{U}_1^{n}$, $m\in\{0,...,D-1\}$ are linearly independent. For any $m$ we have 
$$\hat{U}_0^{D^2-mD-n}=\textnormal{diag}(e^{(D^2-mD-n)i},e^{2(D^2-mD-n)i},...,e^{D(D^2-mD-n)i}).$$
Hence, to assess linear independence, it is sufficient to consider the Vandermonde-type determinant
$$\det\begin{pmatrix}e^{(D^2-n)i} & e^{(D^2-D-n)i} & \cdots & e^{(D^2-(D-1)D-n)i}\\
e^{2(D^2-n)i} & e^{2(D^2-D-n)i} & \cdots & e^{2(D^2-(D-1)D-n)i}\\
\vdots & \vdots & \cdots & \vdots\\
e^{D(D^2-n)i}& e^{D(D^2-D-n)i} & \dots & e^{D(D^2-(D-1)D-n)i}\end{pmatrix}=e^{\frac{D(D+1)}{2}(D^2-n)i}\prod_{1\leq k<l\leq D}\left(e^{-lDi}-e^{-kDi}\right),$$
which is never zero.

%This observation remains valid if $r>2$ $\{(U_0)^i(U_1)^j\}_{i,j=0,...,D-1}$

\end{proof}
We know that $\mathfrak{P}$ and $\mathfrak{I}$ constitute {path-connected} subsets of $\mathfrak{T}$. The reason is that any map in the interior of $\mathfrak{T}$ is primitive, so that any two irreducible maps can be connected by a primitive path through the interior of $\mathfrak{T}$. One core point of this work is to show that any two irreducible maps located at the boundary of $\mathfrak{T}$ can be connected by an irreducible and analytic path \emph{within} the boundary. More precisely,
we have the following theorem.
%Consider the map  with unitary $U_i$. We want to choose $\{U_i\}_{i=1,...,r}$ so that the set of matrix products $\{U_{i_1}\cdot U_{i_2}\cdot ...\cdot U_{i_n}\}_{i_k=1,...,r}$ spans the entire algebra $\cM_D(\mathbb{C})$.
%
\begin{theorem}\label{mainlemma} Let $\mathfrak{T}_{(r)}$ denote the set of linear completely positive maps of Kraus rank $r$ on $\cM_D(\mathbb{C})$ with $r, D\geq2$. Let $\mathfrak{P}_{(r)}\subseteq\mathfrak{T}_{(r)}$ be the subset of primitive maps in $\mathfrak{T}_{(r)}$ and let $\mathfrak{P}_{(r)}^{(TP)}\subseteq\mathfrak{P}_{(r)}$ be the trace-preserving maps in $\mathfrak{P}_{(r)}$. Similarly, let $\mathfrak{I}_{(r)}\subseteq\mathfrak{T}_{(r)}$ be the subset of irreducible maps and $\mathfrak{I}_{(r)}^{(TP)}\subseteq\mathfrak{I}_{(r)}$ the set of maps that are trace-preserving in addition. The following assertions hold:
\begin{enumerate}
\item $\mathfrak{T}_{(r)}$ is path-connected.
\item $\mathfrak{P}_{(r)}$ is a path-connected, dense and relatively open subset of $\mathfrak{T}_{(r)}$.
\item $\mathfrak{P}_{(r)}^{(TP)}$ is a path-connected, dense and relatively open subset of $\mathfrak{T}_{(r)}^{(TP)}$.
\item $\mathfrak{I}_{(r)}$ is a path-connected, dense and relatively open subset of $\mathfrak{T}_{(r)}$.
\item $\mathfrak{I}_{(r)}^{(TP)}$ is a path-connected, dense and relatively open subset of $\mathfrak{T}_{(r)}^{(TP)}$.
%\item The set $\mathfrak{T}_{(r)}-\mathfrak{P}_{(r)}$ has measure zero.
\end{enumerate}
Furthermore, any two elements $\cT^{(0)}$ and $\cT^{(1)}$ of any of the sets
$\mathfrak{T}_{(r)}$, $\mathfrak{P}_{(r)}$, $\mathfrak{I}_{(r)}$, $\mathfrak{P}_{(r)}^{(TP)}$, $\mathfrak{I}_{(r)}^{(TP)}$ can be connected via an interpolating path $t\mapsto\cT^{(t)}$ within the respective set such that matrix entries of $\cT^{(t)}$ depend \emph{analytically} on $t\in[0,1]$.

\end{theorem}
\begin{proof}
\emph{1.}~Let $\cT^{(j)}(X)=\sum_{i=1}^rA_i^{(j)}X(A_i^{(j)})^\dagger$, $j\in\{0,1\}$ be two CP maps of Kraus rank $r$. Let $\gamma:[0,1]\rightarrow\mathbb{C}$ be a curve in the complex plane with $\gamma(0)=0$ and $\gamma(1)=1$. Consider the path of completely positive maps 
\begin{align}\cT^{(\gamma)}(X)=\sum_{i=1}^rA_i^{(\gamma)}X(A_i^{(\gamma)})^\dagger\quad \textnormal{with}\quad A_i^{(\gamma)}=(1-\gamma)A_i^{(0)}+\gamma A_i^{(1)}. \label{path}\end{align}
We prove that $\gamma$ can always be chosen so that $\{A_i^{(\gamma)}\}_{i=1,...,r}$ are linearly independent along $\gamma$. Let $\hat{A}$ denote the $D^2\times r$ matrix with entries $\hat{A}_{(kl),i}^{(\gamma)}=
\left(A^{(\gamma)}_{i}\right)_{kl}$ and $k,l\in\{1,...,D\}$. Since $\cT^{(0)}$ has Kraus rank $r$, there is an $r\times r$ minor ${M}^{(\gamma)}$ of $\hat{A}$ that is non-zero for $\gamma=0$. So ${M}^{(\gamma)}$ is a non-zero polynomial in $\gamma$ whose zeroes are isolated in the complex plane. Hence, we can always choose the path $\gamma$ such that ${M}^{(\gamma)}\neq0$, which implies that $\textnormal{rank}(\hat{A})=r$ along $\gamma$. Note that the latter property can already be achieved choosing $\gamma(t)= t+i\sum_{k=0}^{l-1}x_kt^k$, where $t\in[0,1]$ and $x_k$ are real coefficients: As $t$ varies in
$[0,1]$ $\gamma(t)$ traces out the graph of a degree-$l$ polynomial in the complex plane. It is clear that $l$ and $x_k$ can be chosen such that the graph does not intersect any of the singular points, where $\cT^{(\gamma)}$ has Kraus rank smaller than $r$. If $\gamma(t)= t+i\sum_{k=0}^{l-1}x_kt^k$ then the Kraus matrices $A_i^{(\gamma)}$ and $(A_i^{(\gamma)})^\dagger$
have entries that are polynomials in $t$. Hence $\cT^{(\gamma)}$ can be represented by a matrix whose entries are polynomials in $t$.

% and $\hat{M}^{(\gamma)}$ be an $r\times r$ minor matrix of $\hat{A}^{(\gamma)}$. The determinant of $\hat{M}^{(\gamma)}$ is a polynomial in $\gamma$ and thus is either constantly zero or has \emph{discrete} zeros in the complex plane. The first is excluded as $\cT^{(0)}$ has Kraus rank $r$. If $\gamma$ happens to run through one of the zeros of the determinant then one can slightly alter $\gamma$ to run around this point.

\emph{2.} If $\cT^{(j)}$ $j\in\{0,1\}$ are primitive we prove that we can choose (a general continuous path) $\gamma$ in \eqref{path} so that $\cT^{(\gamma)}$ stays primitive along $\gamma$. To this end we analyze when the Choi matrix $$\omega\left((\cT^{(\gamma)})^{D^4}\right)=\sum_{i_1...i_{D^4}}
(A_{i_1}^{(\gamma)}\cdot...\cdot A_{i_{{D^4}}}^{(\gamma)}\otimes\id)\proj{\varphi}{\varphi}(A_{i_1}^{(\gamma)}\cdot ...\cdot A_{i_{{D^4}}}^{(\gamma)}\otimes\id)^\dagger
=:\sum_{I}\proj{\psi_{I}}{\psi_{I}}$$ is positive definite. Here we have abbreviated $(A_{i_1}^{(\gamma)}\cdot...\cdot A_{i_{{D^4}}}^{(\gamma)}\otimes\id)\ket{\varphi}=\ket{\psi_{I}}$ with multi-index $I=(i_1,...,i_{D^4})$. The determinant of $X\in M_n(\mathbb{C})$ can be written as $\det{X}=\bracket{\Psi_-}{X^{\otimes n}}{\Psi_-}$ with $\ket{\Psi_-}=\frac{1}{\sqrt{n!}}\sum_{\rho\in S_{n}}(-1)^{\textnormal{sgn}(\rho)}\ket{\rho(1)}\otimes\cdots\otimes
\ket{\rho(n)}$. Exploiting this we find
\begin{align}
\det{\omega\left((\cT^{(\gamma)})^{D^4}\right)}=\bracket{\Psi_-}{\omega\left((\cT^{(\gamma)})
^{D^4}\right)^{\otimes D^2}}{\Psi_-}=\sum_{I_1...I_{D^2}}\abs{\braket{\Psi_-}{\psi_{I_1}\otimes...\otimes\psi_{I_{D^2}}}}^2.\label{deter}
\end{align}
The vectors $\ket{\psi_{I}}$ have entries that are \emph{polynomial} in $\gamma$ so that $\braket{\Psi_-}{\psi_{I_1}\otimes...\otimes\psi_{I_{D^2}}}$ is a polynomial in $\gamma$.
Hence, if for given $\gamma$ there is $t^*\in(0,1)$ so that the determinant vanishes at $\gamma(t^*)$ one can slightly deform $\gamma$ to circumvent this zero. Hence there is a curve $\gamma$ connecting $\cT^{(0)}$ and $\cT^{(1)}$ so that $\det{\omega\left((\cT^{(\gamma)})^{D^4}\right)}>0$ along $\gamma$ and $\cT^{(\gamma)}$ is primitive. Following the argument in \emph{1.}~we can also ensure that the Kraus rank stays fixed along $\gamma$ so that $\mathfrak{P}_{(r)}$ is path-connected. As before~it is sufficient to consider $\gamma(t)= t+i\sum_{k=0}^{l-1}x_kt^k$ with real coefficients $x_k$ to ensure that $\cT^{(\gamma)}$ remains primitive and has Kraus rank $r$. For this choice of $\gamma$ we find that $A_i^{(\gamma)}$ and $(A_i^{(\gamma)})^\dagger$
have entries that are polynomial in $t$. 

We prove that $\mathfrak{P}_{(r)}\subseteq\mathfrak{T}_{(r)}$ is dense. Let $\cT\in\mathfrak{T}_{(r)}$ and $\cE\in\mathfrak{P}_{(r)}$ and consider again the path~\eqref{path} interpolating between the Kraus operators of $\cT$ and $\cE$. Again \eqref{deter} is a sum of squares of polynomials, which is non-zero because $\cE$ is primitive and appropriate choice of $\gamma$ ensures that this path is primitive everywhere but in $\cT$. It follows that in any neighborhood of any $\cT\in\mathfrak{T}_{(r)}$ there is a primitive map of Kraus rank $r$. In particular
$\mathfrak{I}_{(r)}\subseteq\mathfrak{T}_{(r)}$ is also dense.

We prove that $\mathfrak{P}_{(r)}\subseteq\mathfrak{T}_{(r)}$ is open in the subspace topology of $\mathfrak{T}_{(r)}$. The maps $\cT\mapsto\cT^n$ and $\cT\mapsto\omega{(\cT)}$ are continuous. Hence for any $\varepsilon>0$ there is a $\delta$-neighbourhood $\cU_\delta$ of $\cE\in\mathfrak{P}_{(r)}$ with the property that $\Norm{\omega{(\cE^{D^4})}-\omega{(\cT^{D^4})}}{}\leq\varepsilon$ for any $\cT\in\cU_\delta$. This implies that $\cU_\delta\subseteq\mathfrak{P}_{(r)}$.

%In particular $\mathfrak{T}_{(r)}\setminus\mathfrak{P}_{(r)}$ and $\mathfrak{I}_{(r)}\setminus\mathfrak{P}_{(r)}$ have Lebesgue measure zero.

\emph{3.}~%We build paths of trace-preserving maps from the paths constructed in \emph{2.}  %In this case we can define the TP map
%$\tilde{\cT}$ by $\tilde{\cT}^*(X):=\frac{1}{\mu}\sum_iP^{-1/2}A_i^\dagger P^{1/2}XP^{1/2}A_iP^{-1/2}$. If $\cT$ is primitive then $\tilde{\cT}$ is primitive as its
%Kraus operators span the entire algebra iff the Kraus operators of $\cT$ do so. Now consider two primitive TPCP maps $\cE$ and $\cT$.
Let $\cT^{(j)}$ with $j\in\{0,1\}$ be primitive CPTP maps of Kraus rank $r$. Choose $\gamma(t)=t+i\sum_{k=0}^{l-1}x_kt^k$ as in \emph{2.}~so that $\cT^{(\gamma)}$ is a path of primitive CP maps of Kraus rank $r$ connecting $\cT^{(0)}$ with $\cT^{(1)}$. Let $\rho^{(\gamma)}$ denote the full-rank eigenvector of $(\cT^{(\gamma)})^*$ corresponding to its spectral radius $\mu^{(\gamma)}$. As $(\cT^{(\gamma)})^*$ is primitive $\mu^{(\gamma)}$ is non-degenerate so that $\mu^{(\gamma)}$ and $\rho^{(\gamma)}$ can be chosen as analytic functions of $t$~\cite[Chap.~9]{Lax}, \cite{kato}. (It is a well-known fact that for a given square matrix $A(t)$ whose elements depend analytically on a real parameter $t$ and a simple eigenvalue of $A(0)$, one has that for all $t$ in a sufficiently small neighborhood of $t=0$ there is a corresponding eigenvalue and a unique eigenvector that depend analytically on $t$.) The family
$$\tilde{\cT}^{(\gamma)}(X)=\sum_{i=1}^r\tilde{A}_i^{(\gamma)}X(\tilde{A}_i^{(\gamma)})^\dagger\quad\textnormal{with}\quad \tilde{A}_i^{(\gamma)}:=(\mu^{(\gamma)})^{-1/2} (\rho^{(\gamma)})^{1/2}A_i^{(\gamma)}(\rho^{(\gamma)})^{-1/2}$$
constitutes an analytic path of TPCP maps from $\cT^{(0)}$ to $\cT^{(1)}$. To see that it is primitive note that
$$\det{\omega\left((\tilde{\cT}^{(\gamma)})^{D^4}\right)}
=(\mu(\gamma))^{-D^6}\det{(\rho^{(\gamma)}\otimes (\rho^{(\gamma)})^{-T})}\det{\omega\left(({\cT}^{(\gamma)})^{D^4}\right)}>0.$$
(We use the abbreviation $X^{-T}=(X^{-1})^T$.) Moreover, the Kraus rank does not change by the rescaled similarity transformation $A_i^{(\gamma)}\mapsto\tilde{A}_i^{(\gamma)}$ and the entries of a matrix representation of $\cT^{(\gamma)}$ are analytic functions of $t$.

\emph{4.~and 5.} Denseness and path-connectedness follow from points \emph{2.}~and \emph{3.}~together with the inclusions $\mathfrak{P}_{(r)}\subseteq\mathfrak{I}_{(r)}\subseteq\mathfrak{T}_{(r)}$ and $\mathfrak{P}_{(r)}^{(TP)}\subseteq\mathfrak{I}_{(r)}^{(TP)}\subseteq\mathfrak{T}_{(r)}^{(TP)}$.
That the sets are relatively open follows from the analogous arguments in \emph{2.}~and \emph{3.}~using Lemma~\ref{mini} and the fact that $\cT\mapsto\frac{1}{2}(\cT+\cT^2)$ is continuous.

%Let $\cE\in\mathfrak{P}_{(r)}$ and $\cT^{(j)}\in\mathfrak{I}_{(r)}\setminus\mathfrak{P}_{(r)}$ for $j\in\{0,1\}$. We connect $\cT^{(0)}$ with $\cE$ using as in \eqref{path}. Along the respective path the determinant is a polynomial in $\gamma$, which is non-zero as $\cE$ is primitive. Connecting $\cE$ with $\cT^{(1)}$ in the same way establishes a path interpolating between $\cT^{(j)}$. $\mathfrak{I}_{(r)}\subseteq\mathfrak{T}_{(r)}$ is open because the map $\cT\mapsto\frac{1}{2}(\cT+\cT^2)$ is continuous so that Lemma~\ref{mini} and the argument from \emph{2.}~apply.

%\emph{5.} The above discussion applies with minor modification.
\end{proof}

\subsection{Irreducible maps with fixed Kraus rank and peripheral spectrum}
An important point about our proof of path-connectedness of $\mathfrak{I}_{(r)}$ is that for $\cT^{(j)}\in\mathfrak{I}_{(r)}\setminus\mathfrak{P}_{(r)}$ the interpolating path is primitive and thus has trivial peripheral spectrum everywhere but at its endpoints. 
%
%It is natural to ask if this is actually necessary, i.e.~if the set $\mathfrak{I}_{(r)}\setminus\mathfrak{P}_{(r)}$ is path-connected and in case to characterize its connected components. In the sequel focus on sets of CPTP maps as this is the case relevant to our application.
%
%
%
%\begin{lemma}
%$\mathfrak{I}_{(r)}^{(TP)}\setminus\mathfrak{P}_{(r)}^{(TP)}$ is not path-connected.
%\end{lemma}
%
%
%\begin{proof}
%Consider two irreducible CPTP maps $\cT^{(j)}\in\mathfrak{I}_{(r,m_j)}^{(TP)}$ with $j\in\{0,1\}$ of different degree $1<m_0<m_1$. Lemma~\ref{irredprop} says that the peripheral spectrum of $\cT^{(j)}$ is $\Big\{e^{2\pi i \frac{k}{m_j}}\Big\}_{k=1,...,m_j}$. Suppose now that there is a continuous path $\cT^{(\gamma)}\subseteq\mathfrak{I}_{(r)}^{(TP)}\setminus\mathfrak{P}^{(TP)}$ interpolating between $\cT^{(j)}$. Along this path the eigenvalues of $\cT^{(\gamma)}$ trace continuous curves~\cite[p.~45]{Ortega} in the closed complex unit disk.  Consider maps $\cT^{(\gamma)}$ in a neighbourhood of $\cT^{(0)}$, which is chosen so small that there exists for any $k$ a \emph{unique} eigenvalue $\lambda_k^{(\gamma)}$ of $\cT^{(\gamma)}$ with $\abs{e^{2\pi i \frac{k}{m_0}}-\lambda_k^{(\gamma)}}\leq\varepsilon$ and in addition all other eigenvalues of $\cT^{(\gamma)}$ have magnitude smaller than $1-\varepsilon$. By irreducibility, along $\gamma$ the peripheral eigenvalues must be $\{e^{2\pi i \frac{k}{m}}\}_{k=1,...,m}$, where $m$ depends on the particular value of $\gamma$. It follows...
%\end{proof}
%
%
In the following we will investigate if it is possible to construct for two irreducible maps with the same peripheral spectrum an irreducible interpolating path that preserves the peripheral spectrum. In other words we ask if the set of irreducible CPTP maps of fixed degree $m$,
$\mathfrak{I}_{(r,m)}^{(TP)}$ is path-connected. A key role in the characterization of the connected components of $\mathfrak{I}_{(r,m)}^{(TP)}$ is played by the unitary eigenvectors $U=\sum_{k=0}^{m-1}\beta^k P_k$ for the eigenvalue $\beta=e^{2\pi i/m}$ of the corresponding (CPU) adjoint maps, cf.~Lemma~\ref{irredprop}: As it turns out, two elements of $\mathfrak{I}_{(r,m)}^{(TP)}$ are in the same connected component iff their adjoint maps have eigenvectors $U$ for eigenvalue $\beta$ such that the multiplicities $(d_1,...,d_m)$, $d_k=\trace{P_k}$, of eigenvalues of $U$ coincide. Here the ordering of eigenvalues is understood so that $d_k$ is the multiplicity of the eigenvalue $\beta^k$ of $U$.
Note that if $U$ is an eigenvector then so is $\beta^l U$ so that connected components of $\mathfrak{I}_{(r,m)}^{(TP)}$ are characterized by tuples $(d_1,...,d_m)$ of algebraic multiplicities of eigenvalues in $U$ \emph{modulo cyclic permutations}. For any $\cT\in\mathfrak{T}_{(r,m)}^{(TP)}$ we call the equivalence class (under cyclic permutations) of tuples $(d_1,...,d_m)$ the \emph{multiplicity index} of $\cT$.
\begin{theorem}\label{compi}
The set $\mathfrak{I}_{(r,m)}^{(TP)}$ of CPTP maps of Kraus rank $r$ that are irreducible of degree $m$ decomposes into path-connected components. Two maps $\cT^{(j)}\in\mathfrak{I}_{(r,m)}^{(TP)}$, $j\in\{0,1\}$, are contained in the same connected component if and only if $\cT^{(j)}$ have the same multiplicity index. In this case they can be connected via an analytic path.
\end{theorem}

In other words $\cT^{(j)}$ are contained in the same connected component iff $(\cT^{(j)})^*$ have unitary eigenvectors $(\cT^{(j)})^*(U^{(j)})=\beta U^{(j)}$, which have multiplicities of eigenvalues $(d_1^{(j)},...,d_m^{(j)})$ that coincide up to cyclic permutations.
\begin{proof}

Recall that if $\cT\in\mathfrak{I}_{(r,m)}^{(TP)}$ then $\cT^*$ is CPU, has Kraus rank $r$ and is irreducible of degree $m$. So, according to~Lemma~\ref{irredprop}, counting multiplicities $\cT^*$ has $m$ eigenvalues of the form $\beta^k$ with corresponding unitary eigenvectors $U^k$. We begin by showing that the set $\mathfrak{I}_{(r,m)}^{(TP)}$ is not path-connected. Suppose that any pair $\cT^{(j)}\in\mathfrak{I}_{(r,m)}^{(TP)}$ with $j\in\{0,1\}$ can be continuously connected by a curve $\cT^{(\gamma)}\subseteq\mathfrak{I}_{(r,m)}^{(TP)}$ with $\gamma(0)=0$, $\gamma(1)=1$. This implies that there exists a continuous curve of eigenvectors $X^{(\gamma)}$ corresponding to the (constant) eigenvalue $\beta$ of $(\cT^{(\gamma)})^*$~\cite[p.~45]{Ortega}. 
%
%It follows from the Schwarz inequality for completely positive maps and the irreducibility of $\cT^*$ that $(X^{(\gamma)})^*X^{(\gamma)}=(\cT^{(\gamma)})^*((X^{(\gamma)})^*X^{(\gamma)})$, see~\cite[Sec.~4]{specevans}. 
%
%But as the eigenvalue $1$ of $(\cT^{(\gamma)})^*$ is non-degenerate by assumption it holds that $X^{(t)}(X^{(\gamma)})^*=\Norm{X^{(\gamma)}}{}^2\id$. 
%
By Lemma~\ref{irredprop} $\beta$ is non-degenerate and has a unitary eigenvector along $\gamma$. Hence, $V^{(\gamma)}=X^{(\gamma)}/\Norm{X^{(\gamma)}}{}$ is a continuous curve of unitary eigenvectors of $(\cT^{(\gamma)})^*$: $(\cT^{(\gamma)})^*(V^{(\gamma)})=\beta V^{(\gamma)}$. Clearly, any other unitary eigenvector $U^{(\gamma)}$ of $(\cT^{(\gamma)})^*$ corresponding to the eigenvalue $\beta$ satisfies $V^{(\gamma)}=\mu(\gamma)U^{(\gamma)}$ with $\abs{\mu(\gamma)}=1$. From Lemma~\ref{irredprop} we can assume that the unitary $U^{(\gamma)}$ has spectral decomposition $U^{(\gamma)}=\sum_{k=0}^{m-1}\beta^k P_k^{(\gamma)}$ with \emph{discrete} eigenvalues $\beta^k$, so that $V^{(\gamma)}=\sum_{k=0}^{m-1}\mu{(\gamma)}\beta^k P_k^{(\gamma)}$. Along the curve $V^{(\gamma)}$ the eigenvalues of $V^{(\gamma)}$ must change continuously~\cite[p.~45]{Ortega}. If $\cT^{(j)}$ had different multiplicity index, we would obtain a contradiction to the continuity of eigenvalues of $V^{(\gamma)}$.

We proceed by proving that if $\cT^{(j)}\in\mathfrak{I}_{(r,m)}^{(TP)}$ have the same multiplicity index then $\cT^{(j)}$ can be connected by a continuous path in $\mathfrak{I}_{(r,m)}^{(TP)}$. We will see later how this yields an analytic path between $\cT^{(j)}$. We write $\cT^{(j)}(X)=\sum_{i=1}^rA_i^{(j)}X(A_i^{(j)})^\dagger$ with $j\in\{0,1\}$ and construct a continuous path in two steps.
\begin{enumerate}
\item We begin with a unitary rotation of $\cT^{(0)}$ to achieve that the Kraus operator of the rotated map and those of $\cT^{(1)}$ are in Frobenius form with respect to the same basis. By assumption $(\cT^{(0)})^*$ and $(\cT^{(1)})^*$ have unitary eigenvectors $U^{(0)}$ and $U^{(1)}$ with the property that counting multiplicities the spectra of $U^{(0)}$ and $U^{(1)}$ coincide. It follows that there is unitary $W$ so that $U^{(1)}=W U^{(0)} W^\dagger$. Clearly, there is a continuous curve of unitaries $W^{(t)}$ with $W^{(0)}=\id$ and $W^{(1)}=W$, which connects $U^{(0)}$ to $U^{(1)}$. On the level of Kraus operators we consider the path $W^{(t)}A_i^{(0)}(W^{(t)})^*$, which performs a continuous rotation so that $W^{(1)}A_i^{(0)}(W^{(1)})^*$ is in Frobenius form in the same basis as $A_i^{(1)}$, see~\emph{4.}~of Lemma~\ref{irredprop}. We write symbolically $A_i^{(1/2)}:=W^{(1)}A_i^{(0)}(W^{(1)})^\dagger$ for the rotated operators. Note that we can always choose $W$ such that $A_i^{(1/2)}$ and $A_i^{(1)}$ have the same block structure, i.e.~$P_k^{(1)}A_i^{(1/2)}=A_i^{(1/2)}P_{k-1}^{(1)}$. The continuous path from $\cT^{(0)}$ to $\cT^{(1/2)}$, which is induced by $W^{(t)}$ is by construction within a subset of maps of $\mathfrak{I}_{(r,m)}^{(TP)}$, whose multiplicity index is given by $\cT^{(0)}$.
%
%
%
%\item The dimensions of the diagonal $0$-blocks in Frobenius form of $\cT^{(0)}$ are given by multiplicities of the eigenvalues of $U^{(0)}$. To achieve that $A_i^{(1/3)}$ and $A_i^{(1)}$ have the exactly same block structure we perform a permutation of blocks in the Frobenius structure of $A_i^{(1/3)}$. (Note that usually a basis does not carry a numbering, so that the Frobenius structure is not determined by the basis alone.) This is achieved by a continuous permutation of basis vectors: $A_i^{(2/3)}=P^{(t)}A_i^{(1/3)}P^{(t)}$ with $P_k^{(1)}A_i^{(2/3)}=A_i^{(2/3)}P_{k-1}^{(1)}$. \textcolor{red}{Ich denke der Punkt hier ist, dass die Frobeniusform durch die Basis nur bis auf permutation der bloecke festgelegt wird. Aber ich bin schon einverstanden, dass man das als Teil von 1 auffassen kann... Dann muss halt $W$ noch zusaetzliche eigenschaften haben. Es gilt aber nicht fuer jedes $W$.}
%
%
%
\item We connected the Kraus operators $A_i^{(0)}$ to $A_i^{(1/2)}$, which have the same Frobenius form as $A_i^{(1)}$. We proceed by constructing an interpolating path between $A_i^{(1/2)}$ and $A_i^{(1)}$ as in the proof of Theorem~\ref{mainlemma}. Let $\gamma:[0,1]\rightarrow\mathbb{C}$ be a curve in the complex plane with $\gamma(0)=1/2$ (so that $A_i^{(\gamma(0))}=A_i^{(1/2)}$) and $\gamma(1)=1$. Consider the path of CP maps $\cT^{(\gamma)}$ with Kraus operators $A_i^{(\gamma)}=(2-2\gamma) A_i^{(1/2)}+ (2\gamma-1) A_i^{(1)}$.
\begin{itemize}
\item As for primitive maps $\{A_i^{(\gamma)}\}_{i=1,...,r}$ can be linearly dependent only at discrete points in the complex plane. We choose $\gamma$ to circumvent these points,~cf.~Theorem~\ref{mainlemma}.
\item $\cT^{(\gamma)}$ is irreducible if and only if $\bar{\cT}^{(\gamma)}=\frac{1}{2}(\cT^{(\gamma)}+(\cT^{(\gamma)})^2)$ is primitive, cf.~Lemma~\ref{mini}. Observe that the Kraus operators $\{\bar{A}_{i}^{(\gamma)}\}_i$ of $\bar{\cT}^{(\gamma)}$ have entries that are \emph{polynomial} in $\gamma$. As in point \emph{2.}~of Theorem~\ref{mainlemma} we conclude that the determinant $\det{\omega\left((\bar{\cT}^{(\gamma)})^{D^4}\right)}$ is a sum of squares of absolute values of polynomials in $\gamma$, whose zeroes are isolated in the complex plane. As this sum is non-zero we can make sure that $\cT^{(\gamma)}$ is irreducible by appropriate choice of $\gamma$.
\item We prove that ${\cT}^{(\gamma)}$ is irreducible of degree $m$ along $\gamma$. Denote by $P_k^{(1)}$ projectors corresponding to the Frobenius form of $\cT^{(1)}$ i.e.~$P_k^{(1)}A_i^{(1)}=A_i^{(1)}P_{k-1}^{(1)}$ for all $i,k$. By construction we even have that $P_k^{(1)}A_i^{(\gamma)}=A_i^{(\gamma)}P_{k-1}^{(1)}$. Let $\rho^{(\gamma)}$ denote the full-rank eigenvector of $(\cT^{(\gamma)})^*$ corresponding to its largest positive eigenvalue $\mu^{(\gamma)}$. For $l\in\{1,...,m\}$ the matrices $\rho^{(\gamma)}(U^{(1)})^l$ are eigenvectors of $({\cT}^{(\gamma)})^*$ with respective eigenvalues $\mu^{(\gamma)}\beta^l$, which can be checked easily:
\begin{align*}
&({\cT}^{(\gamma)})^*\left(\rho^{(\gamma)}(U^{(1)})^l\right)=
(\cT^{(\gamma)})^*\left(\rho^{(\gamma)}\sum_{k=1}^{m}\beta^{kl} P_k^{(1)}\right)\\
&=(\cT^{(\gamma)})^*\left(\rho^{(\gamma)}\right)\sum_{k=1}^{m}\beta^{kl}P_{k-1}^{(1)}=
\mu^{(\gamma)}\beta^l\rho^{(\gamma)}(U^{(1)})^l
\end{align*}
Hence the peripheral spectrum contains $\{\mu^{(\gamma)}e^{2\pi i\frac{k}{m}}\}_{k=1,...,m}$. Consequently $\cT^{(\gamma)}$ is irreducible of degree at least $m$ but the degree might possibly be an integer multiple $n\geq2$ of $m$. To exclude the latter option we consider the map
$$(\cT^{(\gamma)})^m(X)=\sum_{i_1,...,i_m}^rA_{i_1}^{(\gamma)}\cdot...\cdot A_{i_m}^{(\gamma)}X(A_{i_1}^{(\gamma)}\cdot...\cdot A_{i_m}^{(\gamma)})^\dagger$$
and note that the relation $P_k^{(1)}A_i^{(\gamma)}=A_i^{(\gamma)}P_{k-1}^{(1)}$ implies $P_kA_{i_1}^{(\gamma)}\cdot...\cdot A_{i_m}^{(\gamma)}=A_{i_1}^{(\gamma)}\cdot...\cdot A_{i_m}^{(\gamma)}P_k$. Thus the Kraus operators $A_{i_1}^{(\gamma)}\cdot...\cdot A_{i_m}^{(\gamma)}$ have $m$ blocks on the main diagonal corresponding to projectors $P_k$. These blocks constitute Kraus operators of primitive maps iff $\cT^{(\gamma)}$ is irreducible of degree $m$, see~Lemma~\ref{mini}.
%We now follow a similar reasoning as in the proof of \emph{2.}~Theorem~\ref{mainlemma}. 
We can choose a path $\gamma$ so that $(\cT^{(\gamma)})^m$ has $m$ primitive blocks along $\gamma$. To see this we analyze when the restrictions $(\cT^{(\gamma)})^m|_{P_k\cM_D(\mathbb{C})P_k}$ are primitive for any $k$,
i.e.~when the product of determinants (see proof of Theorem~\ref{mainlemma})
$$f(\gamma)=\prod_{k=1}^m\det\omega\left(\left((\cT^{(\gamma)})^m|_{P_k\cM_D(\mathbb{C})P_k}\right)^{D^4}\right)$$
stays positive along $\gamma$. By assumption
$\cT^{(1/2)}$ and $\cT^{(1)}$ are of degree $m$ so that $f(1/2),\ f(1) >0$. As before we find that the occurring determinants are sums of absolute values of polynomials in $\gamma$. Hence, $f(\gamma)$ is a non-zero product of sums of absolute values of polynomials whose zeros are isolated in the complex plane. If for given $\gamma$ there is a $t^*\in(0,1)$ with $f(\gamma(t^*))=0$ one can slightly deform $\gamma$ to circumvent this zero.
In conclusion there is a curve $\gamma$ connecting $\cT^{(1/2)}$ and $\cT^{(1)}$ so that $f(\gamma)>0$ along $\gamma$ and $\cT^{(\gamma)}$ is of degree $m$.

%     By Lemma~\ref{irredprop} there are orthogonal projectors $\{Q_k\}_{k=1,...,nm}$ with the property that $Q_kA_i^{(z)}=A_i^{(z)}Q_{k-1}$. By construction the operators $A_i^{(\gamma)}$ carry a certain Frobenius structure corresponding to the commutation relation
%$P_k^{(1)}A_i^{(\gamma)}=A_i^{(\gamma)}P_{k-1}^{(1)}$ for any point on $\gamma$. The additional intertwining relation with the projectors $Q_k$ corresponds to a further substructure in the Frobenius form at $z$. (It follows from the non-degeneracy of the eigenvalue $e^{\frac{2\pi i}{m}}$ that the unitary $V=\sum_{k=1}^{nm}e^{\frac{2\pi ik}{mn}}Q_k$ satisfies $V^n=U^{(1)}$, which implies that the corresponding spectral projectors satisfy $P_k=\sum_{j=0}^{n-1}Q_{jm+k}$.) Since $\cT^{(0)}$ and $\cT^{(1)}$ are of degree $m$ it follows that \emph{any} deformation of $\gamma$ away from $z$ destroys the particular substructure (localization of zeros) in $A_i$ leading to the intertwining relation with $Q_k$, while maintaining the $P_k$-intertwining relation. Hence, we can achieve that the degree is $m$ by appropriate choice of $\gamma$.

%
%
%
%
%
%
%
%
%
%
%
%
%
%
\item Let $\rho^{(\gamma)}$ denote the full-rank eigenvector of $(\cT^{(\gamma)})^*$ corresponding to its largest positive eigenvalue $\mu^{(\gamma)}$. The family $$\tilde{\cT}^{(\gamma)}(X)=\sum_{i=1}^r\tilde{A}_i^{(\gamma)}X(\tilde{A}_i^{(\gamma)})^\dagger\quad\textnormal{with}\quad \tilde{A}_i^{(\gamma)}:=(\mu^{(\gamma)})^{-1/2} (\rho^{(\gamma)})^{1/2}A_i^{(\gamma)}(\rho^{(\gamma)})^{-1/2}$$ constitutes a continuous path of CPTP maps of Kraus rank $r$. It is irreducible because 
\begin{align*}
&\omega\left(\left(\frac{1}{2}\tilde{\cT}^{(\gamma)}+\frac{1}{2}(\tilde{\cT}^{(\gamma)})^2\right)^{D^4}\right)\geq\\
&\mu_{max}^{-2D^4}\cdot(\rho^{(\gamma)}\otimes(\rho^{(\gamma)})^{-T})^{1/2}\omega\left((\bar{\cT}^{(\gamma)})^{D^4}\right)(\rho^{(\gamma)}\otimes(\rho^{(\gamma)})^{-T})^{1/2}
\end{align*}
with $\mu_{max}:=\max_{\gamma}\mu^{(\gamma)}$ and we know that $\omega\left((\bar{\cT}^{(\gamma)})^{D^4}\right)$ is positive definite. Observe that since $\rho(\tilde{\cT}^{(\gamma)})=\frac{1}{\mu^{(\gamma)}}\rho({\cT}^{(\gamma)})$, the degree of $\tilde{\cT}^{(\gamma)}$ always equals the degree of ${\cT}^{(\gamma)}$.

\end{itemize}
Finally, we explain how to construct an analytic path between a pair of maps $\cT^{(j)}\in\mathfrak{I}_{(r,m)}^{(TP)}$ that have the same multiplicity index.
We consider $\gamma$ of the form $\gamma(t)= t+i\sum_{k=0}^{l-1}x_kt^k$ with real coefficients $x_k$. We choose $x_k$ and $l$ so that $\gamma(0)=1/2$ and $\gamma(1)=1$ and the Kraus operators defined by $A_i^{(\gamma)}=(2-2\gamma) A_i^{(1/2)}+ (2\gamma-1) A_i^{(1)}$ yield a path of CP maps, which is irreducible of degree $m$, has Kraus rank $r$, and is entry-wise analytic, see \emph{3.}~in the proof of Theorem~\ref{mainlemma} and the points above. As before the eigenvalue $\mu^{(\gamma)}$ and the corresponding eigenvector $\rho^{(\gamma)}$ can be chosen as entry-wise analytic functions of $t$~\cite[Chapt.~9]{Lax}, so that the Kraus operators $\tilde{A}_i^{(\gamma)}=(\mu^{(\gamma)})^{-1/2} (\rho^{(\gamma)})^{1/2}A_i^{(\gamma)}(\rho^{(\gamma)})^{-1/2}$
yield an analytic path in $\mathfrak{I}_{(r,m)}^{(TP)}$ connecting $\cT^{(1/2)}$ to $\cT^{(1)}$. To connect $\cT^{(0)}$ to $\cT^{(1)}$ within $\mathfrak{I}_{(r,m)}^{(TP)}$, we simultaneously perform the unitary rotation $W$ and the interpolation $\tilde{A}_i^{(\gamma)}$, i.e.~we consider the path
$(W^{(1-t)})^\dagger \tilde{A}_i^{(\gamma)}W^{(1-t)}$. Here we set $W^{(t)}=e^{iHt}$ with suitable $H$, so that $W^{(0)}=\id$, $W^{(1)}=e^{iH}=W$ and $W^{(t)}$ is component-wise analytic in $t$. Clearly the constructed path has the required endpoints, is analytic in $t$ and is contained in $\mathfrak{I}_{(r,m)}^{(TP)}$.
\end{enumerate}

\end{proof}
\section{Quantum Phases}\label{sec:phases}
%
%The study of quantum phases goes along with the discussion of the large-scale behavior of the corresponding systems. In the context of MPS it turns out that this is linked to the study of spectral properties of CPTP maps and the long-term behaviour of the corresponding quantum Markov chains~\cite{MPSwerner}.
In this section we apply our insights about connectedness of irreducible and primitive CP maps to classify 1D systems with gapped MPS ground state subspace. The following subsection lays down the required framework.
\subsection{Matrix Product States}\label{prel:MPS}
We consider a finite subset $\Lambda\subseteq\mathbb{Z}$ consisting of $N$ contiguous sites, whose Hilbert spaces are each of dimension $r$. Every pure state of this quantum spin system can be written as 
\begin{align*}
\ket{\Psi}=\sum_{i_1,...,i_N}^r\trace{A_{i_1}^{[1]}\cdot A_{i_2}^{[2]}\cdot...\cdot A_{i_N}^{[N]}}\ket{i_1...i_N}
\end{align*}
with site dependent $D_k\times D_{k+1}$ matrices $A_{i_k}^{[k]}$ \cite{Vidal03,MPS}. States of this structure are called Matrix product states (MPS). The integers $D_k$ are referred to as the bond dimension of the MPS. In the case of periodic boundary conditions and translational invariance the matrices can be chosen in a site-independent way~\cite{MPS}, i.\:e.\:
\begin{align*}
\ket{\Psi}=\sum_{i_1,...,i_N}^r\trace{A_{i_1}\cdot A_{i_2}\cdot...\cdot A_{i_N}}\ket{i_1...i_N}
\end{align*}
with $D\times D$ matrices $\{A_i\}_{i=1,...r}$. It follows that a set of matrices $\{A_i\}_{i=1,...,r}$ provides complete description of translationally invariant MPS with periodic boundary conditions. It is an important conceptual step to associate a CP map $\cT(X)=\sum_{i=1}^rA_iXA_i^\dagger$ to such MPS. Obviously the correspondence between maps $\cT$ and MPS is not bijective; for example the
set of Kraus operators $\{XA_i X^{-1}\}_{i=1,...,r}$ with invertible $X$ belongs to the same $\ket{\Psi}$ as above. The following lemma provides a \emph{canonical form}.

\begin{lemma}[\cite{MPSwerner}]\label{canonrep}
Let $\ket{\Psi}$ be a translationally invariant MPS with periodic boundary conditions. The matrices ${A_i}$ can be decomposed as
\begin{align*}
A_i=
 \begin{pmatrix}
  \lambda_1A_i^{(1)} & 0 & \cdots & 0 \\
  0 & \lambda_2A_i^{(2)} & \cdots & 0 \\
  \vdots  & \vdots  & \ddots & \vdots  \\
  0 & 0 & \cdots & \lambda_lA_i^{(b)}
 \end{pmatrix},
\end{align*}
where $\lambda_i\in(0,1]$ and the matrices $A_i^{(j)}$ in the $j$-th block satisfy:

 \emph{1)} The maps $\cT_j(X)=\sum_iA_i^{(j)}X\left(A_i^{(j)}\right)^\dagger$ are unital and irreducible.
 
 \emph{2)} The Kraus operators $A_i^{(j)}$ of $\cT_j$ are given in Frobenius form.
\end{lemma}
Observe that if the map $\cT$ associated with $\ket{\Psi}$ is primitive then there is only one block in the above canonical form and the Frobenius structure of this block is trivial. If $\cT$ has one irreducible block of degree $m>1$ this leads to a decomposition of $\ket{\Psi}$ into a number $m$ of $m$-periodic states~\cite[Thm.~5]{MPS}: We have $\trace{A_{i_1}\cdot A_{i_2}\cdot...\cdot A_{i_N}}=\sum_{k=1}^m\trace{P_kA_{i_1}\cdot A_{i_2}\cdot...\cdot A_{i_N}}$ and exploiting \emph{4.}~in Lemma~\ref{irredprop} it follows that $\ket{\Psi}=\sum_{k=1}^m\ket{\Psi_k}$, where $\ket{\Psi_k}$ have site-dependent matrices $(A_k^{[l]})_{i}=P_{k+l-1}A_{i} P_{k+l}$. Reducible $\cT$ can be assumed to have Kraus operators with the block structure of Lemma~\ref{canonrep}. This leads to a decomposition $\ket{\Psi}=\sum_{a=1}^b\lambda_a\ket{\Psi_{A^{(a)}}}$ and each state $\ket{\Psi_{A^{(a)}}}$ admits further decomposition into periodic components. In this sense an MPS characterizes via its canonical form a \emph{set} of states. This decomposition was originally derived in the context of \emph{finitely correlated states (FCS)} on infinite spin chains, see~\cite[Cor.~3.4]{MPSwerner}.

\subsection{MPS as ground states of Hamiltonians}\label{hamsec}
By an interaction $h$ on an interval $[K,L]\subseteq\Lambda$ of range $L-K+1$ we mean a positive semi-definite operator $h\in\cM_D(\mathbb{C})^{\otimes L-K+1}$. The Hamiltonian for a periodic spin system $\Lambda$ is given by the formal expression $$H_\Lambda=\sum_{i=1}^{N}\tau^i(h),$$
where $\tau$ denotes the translation operation by one site and $h$ is extended to $(\mathbb{C}^r)^{\otimes N}$ by tensoring an implicit identity. Given a translationally invariant MPS $\ket{\Psi}=\sum_{i_1...i_N}^r\trace{A_{i_1}\cdot...\cdot A_{i_N}}\ket{i_1...i_N}$ the so-called \emph{parent Hamiltonian}~\cite{MPSwerner,MPS} model corresponds to a particular choice of the interaction $h$, which ensures that $\ket{\Psi}$ is a ground state of $H_\Lambda$.
More precisely, the parent Hamiltonian model is constructed as follows. For a fixed interval $[K,L]\subseteq\Lambda$ and a given subspace $\cS_{[K,L]}\subseteq(\mathbb{C}^r)^{\otimes K-L+1}$ we denote by $h_{\cS_{[K,L]}}$ the projector onto the orthogonal complement of  $\cS_{[K,L]}$. Of particular interest will be the spaces $\cG_{[K,L]}^{\{A_i\}}\subseteq(\mathbb{C}^r)^{\otimes K-L+1}$, which are spanned by the vectors%
$$\ket{\Psi(Y)}=\sum_{i_K...i_L}^r\trace{YA_{i_K}\cdot...\cdot A_{i_L}}\ket{i_K...i_L},$$
where $Y$ are complex $D\times D$ matrices. We write shortly $\cG$ if the interval and the set of matrices are clear.  

\emph{1.~One primitive block:} If $\ket{\Psi}$ has one primitive block in the canonical form, Lemma~\ref{canonrep}, the parent Hamiltonian is $h=h_{\cG_{[K,L]}}$. We have $H_\Lambda\ket{\Psi}=0$ by construction. Moreover, if the corresponding channel has maximal Kraus rank after $A_0\in\mathbb{N}$ iterations (i.e.~$\omega(\cT^{A_0})>0$) then $\ket{\Psi}$ is the unique ground state of $H_\Lambda$ if $N\geq2A_0$ and $L-K+1>A_0$, see \cite[Thm.~10]{MPS}. Furthermore $H_\Lambda$ has a spectral gap $\delta>0$ above the ground state energy, $H_\Lambda |_{(\mathbb{C}^r)^{\otimes L-K+1}\ominus\ket{\Psi}}\geq\delta\id$, even in the limit of an \emph{infinite} chain, see~\cite{MPSwerner,Nacht1} for a rigorous discussion of this point. A further important point about the parent Hamiltonian construction is that if the dimension of the spaces $\cG_{[K,L]}$ is fixed then the projectors $h_{\cG_{[K,L]}}$ depend {analytically} on the MPS representation of $\ket{\Psi}$ and the CPTP $\cT$. (For sufficiently large $\Lambda$ and $[K,L]$ the space $\cG_{[K,L]}$ is the kernel of the density matrix $\rho_{[K,L]}=\tr_{\Lambda-[K,L]}{\proj{\Psi}{\Psi}}$, whose matrix entries depend analytically on those of $\cT$, cf.~\cite[Chapt.~9]{Lax}, \cite[Lem.~5]{wirperturbations} and Section~\ref{classy} for details.)

\emph{2.~Multiple primitive blocks:} Suppose $\ket{\Psi}$ has $b$ primitive blocks in its canonical form (Lemma~\ref{canonrep}), i.e.~$\ket{\Psi}$ decomposes into states $\ket{\Psi_{A^{(a)}}}$
with $a=1,...,b$. The corresponding parent Hamiltonian is given by ${H}=\sum_i\tau^{i}({h_{\cS_{[K,L]}}})$, where $\cS_{[K,L]}=\bigoplus_{a=1}^b\cG_{[K,L]}^{\{A^{(a)}\}}$. If $[K,L]$ and $\Lambda$ are large enough\footnote{The size of $[K,L]$ and $\Lambda$ depends on $b$ and $A_0=\min\{l \ | \ \omega(\cT_a^{l})>0\ \forall a\}$. Moderate upper bounds are provided in~\cite{MPS}.} then the ground state subspace of $H$ is $\textnormal{span}_{a=1,...,b}\{\ket{\Psi_{A^{(a)}}}\}$~\cite[Thm.~11, 12]{MPS}.

\emph{3.~One irreducible block of degree $m>1$:} If $\ket{\Psi}$ has one irreducible block of degree $m>1$ and $m$ is a factor of $N$ this entails a decomposition $\ket{\Psi}=\sum_{q=1}^{m}\ket{\Psi_q}$ into $m$ $m$-periodic states $\ket{\Psi_q}$. If, however, $m$ is not a factor of $N$ then
$\ket{\Psi}=0$, see~\cite[Theorem 5]{MPS}. In the first case blocking $m$ sites of the chain leads to the situation of \emph{2}, where one has $m$ primitive blocks (see Lemma~\ref{mini}) in the canonical form and the above discussion applies. In other words ${h_{\cS_{[K,L]}}}$ is the local parent Hamiltonian for the $m$-blocked chain and leads to the global Hamiltonian ${H}=\sum_i\tau^{i}({h_{\cS_{[K,L]}}})$, which is translationally invariant with respect to blocked sites. However, with respect to the original unblocked chain local Hamiltonians act on intervals $[(K-1)m+1,Lm]$ and $H$ is not translationally invariant. Translational invariance of ${H}$ is achieved by adding local interactions $h_{\cS_{[K,L]}}$ in a way that does not affect the ground states, see~\cite[Sect.~5]{Nacht1} for details.

\emph{4.~General case:} The parent Hamiltonian construction can be extended to the general situation when $\ket{\Psi}$ admits decomposition into a number of irreducible components $\ket{\Psi_{A^{(i)}}}$ each of which is an average of $m_i$ $m_i$-periodic states $\ket{\Psi_{A^{(i)}}}=\frac{1}{m_i}\sum_{q=1}^{m_i}\ket{\Psi_{A^{(i)},q}}$. Note that the respective component is non-zero only if $m_i$ is a factor of $N$. We will not go into details of this construction but its structure should be clear from \emph{1-3}; we refer to~\cite{ Nacht1,MPSwerner} for details.

% To describe this construction suppose for now that $\ket{\Psi}$ has $b$ primitive blocks in the canonical form. 

%The relevant results were established in~\cite{MPSwerner, Nacht1, MPS} in the context of \emph{Finitely Correlated States/Generalized Valence Bond Solid} states on infinite spin chains. 
%

%
%

\subsection{Equivalence classes of gapped parent Hamiltonians}

%In this section we define what we mean when we say that two Hamiltonians share the same quantum phases.

%Intuitively two Hamiltonians share the same quantum phase, if they can be connected along a smooth path of gapped Hamiltonians. Along such a path the expectation value of any local observable should change continuously and the requirement that the energy gap remains open translates into a smooth behaviour of the ground state energy~\cite{phases,phasesNi,phasesBa,Bach}. 
%
%
\begin{definition}\label{def:phase} Let $H_\Lambda^{(0)}=\sum_i\tau^i(h^{(0)})$ and $H_\Lambda^{(1)}=\sum_i\tau^i(h^{(1)})$ be
translationally invariant and gapped parent Hamiltonians with periodic boundary conditions on $\Lambda$. We say that $h^{(0)}$ and $h^{(1)}$ are \emph{equivalent} iff (after a possible reordering of terms $h$ to reach the same interaction lengths) there is a translationally invariant local Hamiltonian path $H_\Lambda^{(t)}=\sum_i\tau^i(h^{(t)})$, $t\in[0,1]$ with periodic boundary conditions so that:
\begin{enumerate}
\item For all $t\in[0,1]$ and sufficiently large $\Lambda$, the Hamiltonian $H_\Lambda^{(t)}$ has an energy gap above the set of ground states that remains bounded away from zero in the limit $|\Lambda|\rightarrow\infty$ .
\item $h^{(t)}$ is an analytic function of $t\in[0,1]$ and has endpoints $h^{(0)}=h^{(t=0)}$ and $h^{(1)}=h^{(t=1)}$.
\item Local interactions have bounded strength: $\Norm{h^{(t)}}{}\leq1$ in operator norm.
\end{enumerate} 
\end{definition}
Quantum phase transitions are studied on the \emph{infinite} spin chain. As outlined in Section~\ref{hamsec} degeneracy and structure of ground states of $H_\Lambda$ depend on the size of $\Lambda$ and on the interaction range of $h$. For this reason our equivalence relation is defined on local interactions. An important point to note is that we allow reordering and renaming of local interactions.
We will make use of this freedom to achieve the same interaction range in $h^{(0)}$ and $h^{(1)}$ and to choose the interaction range so large that $h^{(t)}$ have the same ground state structure for $t\in[0,1]$. Similar definitions are commonly used to characterize a quantum phase.  As compared to the preceding contribution~\cite{phases} our method does not rely on \emph{blocking}: In~\cite{phases} $k$ physical sites of $\Lambda$ are first blocked into one super-site of dimension $r\cdot k$ and a gapped Hamiltonian path is subsequently constructed on the blocked chain, which is chosen so that only primitive blocks occur in the canonical form of Lemma~\ref{canonrep}. The blocking procedure thus hides the decomposition of MPS into periodic states. In particular, breaking of translational symmetry and associated detection of N\'{e}el order in the MPS/ FCS is lost by blocking~\cite{MPSwerner,Nacht1}. Moreover, Hamiltonian paths obtained after blocking $k$ sites are merely $k$-periodic with respect to the unblocked chain. 
%
%
%
%
%\subsubsection{Robustness and Symmetries}

In addition to the points of Definition~\ref{def:phase} it is often asked that a quantum phase is stable, i.e.~it should constitute an open set in the space of Hamiltonians~\cite{phases}. In other words in the thermodynamic limit of large $\Lambda$ the perturbed Hamiltonian
$$H'_\Lambda=H_\Lambda+\sum_{i=1}^N\tau^i(\phi)$$
should have a positive energy gap above the ground state when the perturbation $\phi$ is small enough. Robustness of the parent Hamiltonian model with unique ground state was established in~\cite{Stable,Rob} and~\cite[Thm.~2]{wirperturbations}. Hamiltonians with degenerate MPS ground states do not satisfy this property: Different $\ket{\Psi_{A_i}}$ are locally supported on linearly independent spaces and a local perturbation that is chosen in the direction of one of these spaces will change the ground space degeneracy of $H_\Lambda$. In this case stability can be achieved by symmetry protection, i.e.~by imposing a symmetry that permutes ground states such that perturbation affects all ground states simultaneously~\cite{phasesNi,phases}.

\subsection{Classification of gapped parent Hamiltonians}\label{classy}
%
%
%The classification of quantum phases of systems with MPS ground state subspace depends on whether or not the bond dimensions of the ergodic components are kept fixed or are allowed to change.
%\begin{example} We consider a spin chain $\Lambda$ with sites of physical dimension $r$ and two MPS $\ket{\Psi^{(j)}}$, $j\in\{0,1\}$, with bond dimension $D$ and one primitive block in the canonical form, cf.~Lemma~\ref{canonrep}. Let $H_\Lambda^{(j)}$ be the corresponding parent Hamiltonians, which have $\ket{\Psi^{(j)}}$ as there unique gapped ground states and let $\cT^{(j)}$ denote the associated CPTP maps of Kraus rank $r$. To assess whether $H_\Lambda^{(j)}$ are in the same quantum phase, we proceed as follows: We regroup local terms in $H_\Lambda^{(j)}$ so that $H_\Lambda^{(j)}$ have the same interaction length. We employ Theorem~\ref{mainlemma} to construct a primitive path $\cT^{(\gamma)}$ of CPTP of Kraus rank $r$ that interpolates between $\cT^{(j)}$. The Kraus rank hat to remain fixed to ensure that the resulting path of MPS in fact can be interpreted as a path of states on $\Lambda$. The primitivity of $\cT^{(\gamma)}$ guarantees that along this path the ground state remains non-degenerate and the spectral gap remains open even when $\Lambda$ gets large. We conclude that the resulting Hamiltonian path is in accordance with Definition~\ref{def:phase} so that $H_\Lambda^{(j)}$ share the same phase. 
%\end{example}
It seems generally expected that two parent Hamiltonians share the same phase iff they have same degeneracies of ground states, see~\cite{phases} as well as the recent article~\cite{Bach}. Here, we provide sufficient conditions for equivalence of two given local parent Hamiltonians $h^{(0)}$ and $h^{(1)}$. Our method relies on the following idea: Suppose a family of primitive CPTP maps $\cT^{(t)}$ continuously interpolates between maps $\cT^{(0)}$ and $\cT^{(1)}$. For any $t$ we can consider a parent Hamiltonian $h^{(t)}$ formed from the Kraus operators of $\cT^{(t)}$. As $\cT^{(t)}$ is primitive this Hamiltonian is gapped and it should depend analytically on $t$ when $\cT^{(t)}$ is analytic in $t$. In our analysis a crucial role is played by the Kraus rank and the canonical form of the CPTP maps $\cT^{(t)}$. The Kraus rank determines the \emph{physical dimension} of the local sites of the spin chain. The canonical form reflects \emph{the structure of the ground state subspace} and is crucial for the gappedness of the Hamiltonian path. 

%It is generally expected that any two systems with MPS ground states that share the same ground state subspace degeneracy are in the same phase, see Theorem~\ref{frei} and~\cite{phases}, as well as the recent article~\cite{Bach}. We prove that this is the case if gapped phases are studied with respect to Definition~\ref{def:phase}. 

%On the other hand if the bond dimension is kept fixed along this path, see Subsection~\ref{interbnd}, then the phase portrait of systems with MPS ground states is characterized by the canonical form of ground states.
%
%

%
%
\begin{theorem}\label{final}Let $\ket{\Psi^{(j)}}$ with $j\in\{0,1\}$ be translationally invariant MPS and let $h^{(j)}$ be the corresponding parent Hamiltonians as in Section~\ref{hamsec}. If $\ket{\Psi^{(j)}}$ have identical structure and block dimensions in the canonical form, Lemma~\ref{canonrep}, then $h^{(j)}$ are equivalent according to Definition~\ref{def:phase}. %If $H_\Lambda^{(j)}$ have manifolds of ground states have that have different dimension, then $H_\Lambda^{(j)}$ are in different phases.
%
%\begin{enumerate}
%\item At variable bond dimension $H_\Lambda^{(j)}$ with $j\in\{0,1\}$ are in the same phase iff their manifolds of ground states have the same dimension.
%We have that 
%\end{enumerate}
%
\end{theorem}
The above implies that for sufficiently large $\Lambda$ we can interpolate between $H_\Lambda^{(j)}$ with $j\in\{0,1\}$ using a gapped path. On the other hand, if the ground state subsets of $H_\Lambda^{(j)}$ with $j\in\{0,1\}$ have different dimensions, then $H_\Lambda^{(j)}$ cannot be connected via a continuous \emph{gapped} path. 
\begin{proof}[Proof of Theorem~\ref{final}]
An important point related to the parent Hamiltonian model is that even in the primitive case we know that the global Hamiltonian $H_\Lambda$ has unique gapped ground state only if the local interaction range is large enough and $\Lambda$ has enough sites, see Section~\ref{hamsec} for sufficient conditions. We avoid this problem in our proofs by allowing to group and rename local Hamiltonians to reach large enough interaction range, e.g.~in the primitive case it is sufficient to group local Hamiltonians so that they have interaction range $D^4+1$.
If $\ket{\Psi^{(j)}}$ have identical structure and block dimensions in the canonical form then we construct an analytic path interpolating between $h^{(j)}$ as follows.

\begin{enumerate}
\item \emph{One primitive block of dimension $D$:} Suppose $\ket{\Psi^{(j)}}$ have one primitive block in the canonical form and share the same bond dimension $D$. We consider the parent Hamiltonian construction with spaces $\cG_{[1,L]}^{(j)}$ belonging to $\ket{\Psi^{(j)}}$, where we choose $A_0=D^4\geq\max_{j}{A_0^{(j)}}$, cf.~Section~\ref{hamsec}.
We know from Theorem~\ref{mainlemma} that $\mathfrak{P}_{(r)}^{(TP)}$ is path connected, which provides an associated family of spaces $\cG_{[1,L]}^{(t)}$ spanned by vectors
$$\sum_{i_1,...,i_L}^r\trace{YA_{i_1}^{(t)}\cdot...\cdot A_{i_L}^{(t)}}\ket{i_1,...,i_L},\quad Y\in\cM_D(\mathbb{C})$$
interpolating between $\cG_{[1,L]}^{(j)}$. These paths can be chosen such that the matrix entries of $A_{i_1}^{(t)}$ depend analytically on $t$. Note that primitivity implies that the dimension of $\cG_{[1,L]}^{(\gamma)}$ is $D^2$ for any $t\in[0,1]$. The density matrices
\begin{align*}
\rho^{(t)}=\frac{1}{D}\sum_{i_1....i_{L}\atop{j_1...j_{L}}}^{d}
\trace{A^{(t)}_{i_1}\cdot...\cdot A^{(t)}_{i_{L}} (A^{(t)}_{j_{L}})^\dagger\cdot...\cdot (A^{(t)}_{j_{1}})^\dagger}\proj{i_1...i_{L}}{j_1...j_{L}}.
\end{align*}
constitute an analytic path with $\textnormal{Im}(\rho^{(t)})=\cG_{[1,L]}^{(t)}$. As $\rho^{(t)}$ have fixed rank and depend analytically on $t$ it follows from~\cite[Chapt.~9]{Lax} that projectors $h^{(t)}$ onto $\textnormal{ker}(\rho^{(t)})$ form a path that is analytic in $t$. This establishes an analytic path of local Hamiltonians $h^{(t)}$, which is gapped as $\ket{\Psi^{(t)}}$ is primitive. 

\item \emph{Multiple primitive blocks:} Suppose that the states $\ket{\Psi^{(j)}}$ decompose into $b$ primitive blocks as in Section~\ref{hamsec} \emph{2}. We consider the parent Hamiltonian construction with spaces $\cS_{[1,L]}^{(j)}=\bigoplus_{a=1}^b\cG_{[1,L]}^{(j),{\{A^{(a)}\}}}$, where we assume that all spaces $\cG^{(j)}$ have the same dimension $(D^{(j)})^2$. Using Theorem~\ref{mainlemma} we find primitive paths $\cT^{(t)}_a$
interpolating between the maps $\cT^{(j)}_a$ that make up the blocks in $\cT^{(j)}$, cf.~Lemma~\ref{canonrep}. Note that the primitivity of $\cT^{(t)}_a$ can be achieved simultaneously for all $a$ by \emph{one} analytic choice of $\gamma$ as points where $\cT^{(\gamma)}_a$ is not primitive are isolated in the complex plane, see Theorem~\ref{mainlemma}. As in point \emph{1.} above, this yields analytic deformations of spaces $\cG_{[1,L]}^{(t),{\{A^{(a)}\}}}$ interpolating between $\cG_{[1,L]}^{(j),{\{A^{(a)}\}}}$. Hence, to find an interpolating path between $h_{\cS_{[1,L]}^{(j)}}$ we only need to show that the dimension of spaces $\cS_{[1,L]}^{(t)}$ remains fixed. This is guaranteed by~\cite[Prop. 4]{MPS}, which asserts that under the assumption that $L$ is large enough the spaces $\cG^{{\{A^{(a)}\}}}_{[1,L]}$ are linearly independent.

\item \emph{One irreducible block of multiplicity index $(d_1,...,d_m)$:} If $\ket{\Psi^{(j)}}$ have associated CPTP maps $\cT^{(j)}$ that are irreducible and of the same multiplicity index $(d_1,...,d_m)$ with $\sum_id_i=D$, then an analytic interpolating path can be constructed by the following procedure. From Theorem~\ref{compi} we have an analytic path ${\cT}^{(t)}$ that lies in $\mathfrak{I}^{(TP)}_{(r,m)}$ and interpolates between $\cT^{(j)}$. Together with the parent Hamiltonian construction of Section~\ref{hamsec} \emph{3.}~this yields an analytic gapped path $h^{(t)}$ interpolating between
${h}^{(0)}$ and ${h}^{(1)}$. The corresponding ground state is the $m$-periodic state $\ket{\Psi^{(t)}}$. An important point is that to construct this path we could not first block $m$ sites of the chain and then simply employ the discussion of point \emph{3.}~above; it is not clear if the resulting Hamiltonian path could be \lq\lq{}unblocked\rq\rq{} to be interpreted as a path on a spin chain of local dimension $r$. %Nevertheless the Hamiltonians $h^{(t)}$ are constructed by blocking $m$ sites, so that in order to interpolate between ${h}^{(j)}$ and $\tilde{h}^{(j)}$ we can rely on the discussion of point \emph{3.}~as then the ground states have the right properties (physical dimension) by assumption.
\item \emph{Multiple irreducible blocks:} This case is an immediate consequence of points \emph{3.-4.}~above.

\end{enumerate}
\end{proof}

Theorem~\ref{final} does not provide a complete classification of gapped MPS phases. We did not address the situation when the canonical forms of $\ket{\Psi^{(j)}}$ have blocks of \emph{different bond dimension}. For instance if $\ket{\Psi^{(j)}}$ have one primitive block each and different bond-dimensions $D^{(0)}<D^{(1)}$, then the spaces $\cG_{[1,L]}^{(j)}$ have different dimensions $(D^{(0)})^2<(D^{(1)})^2$. Clearly this implies that there is no continuous path of projectors, which interpolates between $h^{(j)}$. In particular the corresponding parent Hamiltonians \textit{cannot be continuously connected by a path of parent Hamiltonians}. Furthermore we did not discuss the converse assertion, i.e.~which conditions on $h^{(j)}$ (beyond the same dimension of the ground state manifold) are necessary to conclude that $h^{(j)}$ are inequivalent. A natural question to ask in this context would be, if the fact that $\mathfrak{I}_{(r,m)}^{(TP)}$ is not path-connected implies that $h^{(j)}$ corresponding to different indices of multiplicity are, in any reasonable sense, not in the same \lq\lq{}quantum phase\rq\rq{}.

\acknowledgements{We would like to thank Sven Bachmann, Toby Cubitt, Christian Neumaier, M{a}ris Ozols and William Matthews for fruitful discussions. We acknowledge financial support from the European Union under project QALGO (Grant Agreement No. 600700). This work was made possible partially through the support of grant
\#48322 from the John Templeton Foundation. The opinions
expressed in this publication are those of the authors and do not
necessarily reflect the views of the John Templeton Foundation.
}
\appendix
\section{Proofs}
\subsection{Equivalent characterization of primitivity}\label{equivalpr}
The equivalence of the points in Definition~\ref{primitive} was established for TP maps in~\cite{Wielandt}. Here, we generalize the equivalence to general CP maps.

To any CP map $\cT$ we can associate a map $\cF(X)=\frac{\cT(X)}{\trace{\cT(X)}}$ for all $X\in\cM_D(\mathbb{C})$. $\cF$ is continuous and maps the compact convex set $\cS=\{\sigma\in\cM_D(\mathbb{C})\: |\: \sigma\geq0,\: \trace{\sigma}=1\}$ of density matrices into itself.
By Brouwer's fixed point theorem there is $\rho\in \cS$ with $\cF(\rho)= \rho$ and thus $\cT(\rho)= \mu\rho$, $\mu=\trace{\cT(\rho)}$. 

Now consider \emph{1.} in Definition~\ref{primitive}. Let $n$ be chosen so that $\omega(\cT^n)>0$ (i.e.~$\omega$ is positive definite). It follows that for any vector $\psi$ we have $$\left(\id\otimes\bra{\psi}\right)\omega(\cT^n)\left(\id\otimes\ket{\psi}\right)=\cT^n\left(\proj{\bar{\psi}}{\bar{\psi}}\right)>0,$$
where $\bar{\cdot}$ denotes complex conjugation. Hence for any nonzero $\rho\geq0$ we have $\cT^n(\rho)>0$. If now the eigenvector $\rho$ with $\cT(\rho)= \mu\rho$ had not full rank we would have $\cT^n(\rho)= \mu^n\rho\ngtr0$, which would contradict the assumption $\omega(\cT^n)>0$. Hence $\rho$ has maximal rank and we can consider the map
$$\tilde{\cT}(\cdot)=\frac{1}{\mu}\rho^{1/2}\cT^*(\rho^{-1/2}\cdot\rho^{-1/2})\rho^{1/2},$$
which is CPTP. It is shown in~\cite[Thm.~1]{Wielandt} that
if there exists $n$ so that $\omega(\tilde{\cT}^n)>0$, then $\omega(\tilde{\cT}^{D^4})>0$. However, as for all $n$ we have that
$$\det(\omega(\tilde{\cT}^n))=\mu^{-nD^2}\det(\rho\otimes\rho^{-T})\det(\overline{\omega(\cT^n)}),$$
it follows $\omega(\tilde{\cT}^n)>0\Leftrightarrow\omega(\cT^n)>0$, which establishes the equivalence between \emph{1.}~and~\emph{2.}.

Moreover, the condition $\omega(\tilde{\cT}^n)>0$ for some $n$ by~\cite[Prop.~3]{Wielandt} is equivalent to $\tilde{\cT}$ having only one eigenvalue $\lambda$ of magnitude $1$ and the corresponding eigenvector being positive definite. As $1$ is an eigenvalue of $\tilde{\cT}$
with $\tilde{\cT}^*(\id)=\id$, we have that $\lambda=1$, the spectral radius of $\tilde{\cT}$ is a non-degenerate eigenvalue with positive left- and right- eigenvectors and all other eigenvalues having strictly smaller magnitude. Note that $\cT$ is related to $\tilde{\cT}$ by a conjugation, a similarity transform and a global stretching by a factor of $\mu$. Hence, $\cT$ has the properties of \emph{3.} if and only if \emph{1.} holds.
\subsection{Proof of Lemma~\ref{mini}}\label{AppB}
\begin{enumerate}
\item
Suppose that $\cT$ is CP and irreducible. Let $\mu$ be the spectral radius of $\cT$ and note (\emph{2.}~in Definition~\ref{irreducible}) that $\mu$ is a non-degenerate eigenvalue of $\cT$. The spectrum of $\bar{\cT}$ is simply $\{\frac{1}{2}(\lambda+\lambda^2)\}_{\lambda\in\rho(\cT)}$. Hence, $\bar{\cT}$ has spectral radius $\frac{1}{2}(\mu+\mu^2)$, which is also an eigenvalue of $\bar{\cT}$. It is easy to check that $\abs{\lambda}\leq\mu$ with $\lambda\neq\mu$ implies $\frac{1}{2}\abs{\lambda+\lambda^2}<\frac{1}{2}(\mu+\mu^2)$. As $\mu$ is non-degenerate
it follows that the spectral radius of $\bar{\cT}$ is non-degenerate and moreover that the peripheral spectrum is trivial. The left- and right- eigenvectors of $\cT$ corresponding to $\mu$ are positive and are also (left-/right-) eigenvectors of $\bar{\cT}$ for $\frac{1}{2}(\mu+\mu^2)$. Thus $\bar{\cT}$ is primitive.

Suppose that $\bar{\cT}$ is primitive. If $\mu$ was a degenerate eigenvalue of $\cT$ then $\frac{1}{2}(\mu+\mu^2)$ would be a degenerate eigenvalue of $\bar{\cT}$, which is impossible. Moreover, the left- and right eigenvectors of $\cT$ corresponding to $\mu$ are also eigenvectors of $\bar{\cT}$ corresponding to $\frac 1 2 (\mu+\mu^2)$. Hence, they have to be positive definite and $\cT$ is irreducible.

\item Suppose that $\cT$ is CP and irreducible. Let $\mu$ be the spectral radius of $\cT$. From \emph{2.}~in Definition~\ref{irreducible} we have that $\mu$ is a non-degenerate eigenvalue of $\cT$ and that the corresponding (right-) eigenvector $\rho$ is positive definite. The map 
$$\tilde{\cT}^*(\cdot)=\frac{1}{\mu}\rho^{-1/2}\cT(\rho^{1/2}\cdot\rho^{1/2})\rho^{-1/2}$$
is irreducible and unital.

Suppose $\cT$ is irreducible of degree $m$ then $\tilde{\cT}^*$ is also irreducible of degree $m$ since we applied a rescaled similarity transformation. It follows from \emph{4.}~in Lemma~\ref{irredprop} that there is a set of $m$ projectors $P_k\neq0$ with $\sum_{k=1}^mP_k=\id$ and $P_k \tilde{A}_i= \tilde{A}_iP_{k-1}$. For the map
$$(\tilde{\cT}^*)^m(X)=\sum_{i=1}^r\tilde{A}_{i_1}\cdot...\cdot\tilde{A}_{i_m}(X)(\tilde{A}_{i_1}\cdot...\cdot\tilde{A}_{i_m})^\dagger$$ this implies that $\tilde{A}_{i_1}\cdot...\cdot\tilde{A}_{i_m}P_k=P_k\tilde{A}_{i_1}\cdot...\cdot\tilde{A}_{i_m}$, which says that the Kraus operators $\tilde{A}_{i_1}\cdot...\cdot\tilde{A}_{i_m}$ can be written with $m$ blocks on the main diagonal corresponding to projectors $P_k$. Since 
$(\tilde{\cT}^*)^m$ is unital the restriction $(\tilde{\cT}^*)^m|_{P_k\cM_D(\mathbb{C})P_k}$ is unital, too. Since $\cT$ is of degree $m$ the map $(\tilde{\cT}^*)^m$ has an $m$-fold degenerate eigenvalue $1$ and all other eigenvalues have strictly smaller magnitude. This implies that $1$ is a non-degenerate eigenvalue of $(\tilde{\cT}^*)^m|_{P_k\cM_D(\mathbb{C})P_k}$ with positive eigenvector $\id$ and all other eigenvalues of $(\tilde{\cT}^*)^m|_{P_k\cM_D(\mathbb{C})P_k}$ have strictly smaller magnitude. To conclude that $(\tilde{\cT}^*)^m|_{P_k\cM_D(\mathbb{C})P_k}$ is primitive it remains to show that it has a positive definite left-eigenvector corresponding to the eigenvalue $1$. Let $\tilde{\sigma}$ denote the positive-definite (right-) eigenvector of $\tilde{\cT}$ corresponding to the eigenvalue $1$.
It follows from the intertwining-relation $P_k \tilde{A}_i= \tilde{A}_iP_{k-1}$ that
$\sum_{k=1}^m P_k\tilde{\sigma}P_k$ is an eigenvector with eigenvalue $1$ as well:
$$\tilde{\cT}(\sum_{k=1}^mP_k\tilde{\sigma}P_k)=\sum_{k=1}^mP_k\tilde{\sigma}P_k.$$
But as $\tilde{\cT}$ is irreducible $1$ is a non-degenerate eigenvalue so that
$\tilde{\sigma}$ is block-diagonal with respect to $P_k$, i.e.~$\tilde{\sigma}=\sum_{k=1}^m P_k\tilde{\sigma}P_k$. It follows that $P_k\tilde{\sigma}P_k$ is a positive definite left-eigenvector for the eigenvalue $1$ of $(\tilde{\cT}^*)^m|_{P_k\cM_D(\mathbb{C})P_k}$, which therefore has to be primitive. Finally, set 
$\tilde{P}_k=\rho^{1/2} P_k \rho^{-1/2}$ and observe that the above discussion implies that $\cT^m(\tilde{P}_k\cM_D{(\mathbb{C})}\tilde{P}_k)\subseteq \tilde{P}_k\cM_D{(\mathbb{C})}\tilde{P}_k$ and that $\cT^m|_{\tilde{P}_k\cM_D(\mathbb{C})\tilde{P}_k}$ is primitive.

We now prove the converse assertion. As before $\cT$ is CP and irreducible. If $\cT$ is irreducible of degree $d>m$ then $\tilde{\cT}^m$ has eigenvalues of magnitude $1$ that are not $1$. For this reason $\cT^m$ cannot be a direct sum of primitive blocks. If $\cT$ is irreducible of degree $d<m$ then $\tilde{\cT}^d$ has $d<m$ primitive blocks. Hence, either does $\tilde{\cT}^m$ have non-trivial eigenvalues of magnitude $1$ or $d<m$ primitive blocks.

% \begin{align*}
%&({\cT}^{(\gamma)})^*\left(\rho^{(\gamma)}(U^{(1)})^l\right)=
%(\cT^{(\gamma)})^*\left(\rho^{(\gamma)}\sum_{k=1}^{m}\beta^{kl} P_k^{(1)}\right)\\
%&=(\cT^{(\gamma)})^*\left(\rho^{(\gamma)}\right)\sum_{k=1}^{m}\beta^{kl}P_{k-1}^{(1)}=
%\mu^{(\gamma)}\beta^l\rho^{(\gamma)}(U^{(1)})^l
%\end{align*}
  
\end{enumerate}

\bibliographystyle{abbrv}

\begin{thebibliography}{10}

\bibitem{albev}
S.~Albeverio, R.~Hoegh-Krohn.
\newblock {Frobenius theory for positive maps of von Neumann algebras}.
\newblock {\em Comm. Math. Phys.}, 64(1):83--94, 1978.


\bibitem{phasesBa}
S.~Bachmann, S.~Michalakis, B.~Nachtergaele, R.~Sims.
\newblock {Automorphic equivalence within gapped phases of quantum lattice systems}.
\newblock {\em Comm. Math. Phys.}, 309(3):835--871, 2012.

\bibitem{Bach2}
S.~Bachmann and B.~Nachtergaele.
\newblock {Product vacua with boundary states}.
\newblock {\em Phys. Rev. B}, 86(3):035149, 2012.

\bibitem{Bach1}
S.~Bachmann and B.~Nachtergaele.
\newblock {Product vacua with boundary states and the classification of gapped
  phases}.
\newblock {\em Comm. Math. Phys.}, 329(2):509--544, 2014.

\bibitem{Bach}
S.~Bachmann and Y.~Ogata.
\newblock {$\cC^1$-Classification of gapped parent Hamiltonians of quantum spin
  chains}.
\newblock 2014.
\newblock arXiv:1407.3924v2.


\bibitem{undecidablegap}
T.~S.~Cubitt, D.~Perez-Garcia, M.~M.~Wolf. 
\newblock {Undecidability of the Spectral Gap}.
\newblock arXiv:1502.04573; arXiv:1502.04135.

\bibitem{phasesNi}
X.~Chen, C.~Gu, and X.~Wen.
\newblock {Classification of Gapped Symmetric Phases in 1D Spin Systems}.
\newblock {\em Phys. Rev. B}, 83:035107, 2011.

\bibitem{Rob}
J.~I.~Cirac, S.~Michalakis, D.~Perez-Garcia, N.~Schuch.
\newblock {Robustness in Projected Entangled Pair States}.
\newblock 2013.
\newblock arXiv: 1306.4003.

\bibitem{Daviesirred}
E.~Davies.
\newblock {Quantum stochastic processes II}.
\newblock {\em Comm. in Math. Phys}, 19:83--105, 1970.

\bibitem{specevans}
D.~Evans and R.~Hoegh-Krohn.
\newblock {Spectral properties of positive maps on C*-algebras}.
\newblock {\em J. of London Math. Soc.}, 17(2):345--355, 1978.

\bibitem{MPSwerner}
M.~Fannes, B.~Nachtergaele, R.~Werner.
\newblock Finitely correlated states on quantum spin chains.
\newblock {\em Comm. in Math. Phys.}, 144(3):443--490, 1992.

\bibitem{farenick}
D.~R. Farenick.
\newblock {Irreducible Positive Linear Maps on Operator Algebras}.
\newblock {\em Proc. Amer. Math. Soc.}, 124:11:3381--3390, 1996.

\bibitem{Has1}
M.~Hastings.
\newblock {An Area Law for One Dimensional Quantum Systems}.
\newblock {\em J. Stat. Mech}, P08024, 2007.
\newblock arXiv:0705.2024.

\bibitem{Has2}
M.~Hastings.
\newblock {Entropy and entanglement in quantum ground states}.
\newblock {\em Phys. Rev. B}, 76:035114, 2007.

\bibitem{kato}
T.~Kato.
\newblock {\em {Perturbation Theory for Linear Operators}}.
\newblock Springer: Classics in Mathematics, 1980.

\bibitem{Lax}
T.~Lax.
\newblock {\em {Linear Algebra}}.
\newblock Pure and Applied Mathematics, Wiley, 1996.

\bibitem{Stable}
S.~Michalakis and J.~Pytel.
\newblock {Stability of Frustration-Free Hamiltonians}.
\newblock {\em Comm. in Math. Phys.}, 322(2):277--302, 2013.

\bibitem{Nacht1}
B.~Nachtergaele.
\newblock {The spectral gap for some quantum spin chains with discrete symmetry
  breaking}.
\newblock {\em Comm. in Math. Phys.}, 175:565--606, 1996.

\bibitem{Ortega}
J.~Ortega.
\newblock {\em {Numerical Analysis: A second course}}.
\newblock SIAM: Classics in Applied Mathematics, 1972.

\bibitem{MPS}
D.~Perez-Garcia, F.~Verstaete, M.~M.~Wolf, and J.~I.~Cirac.
\newblock Matrix product state representations.
\newblock {\em Quantum Inf. Comput.}, 7:401--430, 2007.

\bibitem{Wielandt}
M.~Sanz, D.~P\'{e}res-Garc\'{i}a, M.~M.~Wolf, J.~I.~Cirac.
\newblock A quantum version of {W}ielandt's inequality.
\newblock {\em IEEE Trans. Inf. Theory}, 2010.

\bibitem{phases}
N.~Schuch, D.~Perez-Garcia, and J.~I.~Cirac.
\newblock {Classifying quantum phases using MPS and PEPS}.
\newblock {\em Phys. Rev. B}, 84:165139, 2011.

\bibitem{wirperturbations}
O.~Szehr and M.~M.~Wolf.
\newblock {Perturbation theory for parent Hamiltonians of Matrix product
  states}.
\newblock {\em J. Stat. Phys.}, 159(4):752--771, 2015.

\bibitem{szrewo}
O.~Szehr, D.~Reeb, M.~M.~Wolf.
\newblock {Spectral convergence bounds for classical and quantum Markov processes}.
\newblock {\em J. Comm. in Math. Phys.}, 333(2):565--595, 2015.

\bibitem{Verstr12}
F.~Verstraete and J.~I.~Cirac.
\newblock {Matrix product states represent ground states faithfully}.
\newblock {\em Phys. Rev. B}, 73:094423, 2006.

\bibitem{Vidal03}
G.~Vidal.
\newblock Efficient classical simulation of slightly entangled quantum
  computations.
\newblock {\em Phys. Rev. Lett.}, 91:147902, 2003.

\bibitem{wernerall}
R.~F. Werner.
\newblock {All Teleportation and Dense Coding Schemes}.
\newblock {\em J. Phys.}, 34:7081, 1996.
\newblock arXiv: quant-ph/0003070v1.

\bibitem{Michael:Skript}
M.~M.~Wolf.
\newblock Quantum channels and operations: Lecture notes.
\newblock CTAN:
  \url{http://www-m5.ma.tum.de/foswiki/pub/M5/Allgemeines/MichaelWolf/QChannelLecture.pdf},
  2011.

\bibitem{WOVC}
M.~M.~Wolf, G.~Ortiz, F.~Verstraete, J.~I.~Cirac.
\newblock {Quantum Phase Transitions in Matrix Product Systems}.
\newblock {\em Phys. Rev. Lett.}, 97(11):110403, 2006.

\end{thebibliography}

%
%
%

%LITTER LITTER LITTER
%
%
%
%
%
%We will employ this result to study continuous deformations of MPS with one primitive block in the canonical form. However, Lemma~\ref{canonrep} only guarantees that the CP map associated to a certain block in the canonical form is irreducible. In the following we further analyse connectivity of the set of irreducible (and not primitive) CP maps of fixed Kraus rank. We start with a simple observation.
%
%
%\begin{lemma}\label{compir}
%The set $\mathfrak{I}_{(r)}^{(TP)}$ of irreducible CPTP maps of Kraus rank $r$ is not path-connected. Subsets $\mathfrak{I}_{(r,m)}^{(TP)}\subseteq\mathfrak{I}_{(r)}^{(TP)}$
%of irreducible maps of degree $m$ cannot be connected if $m_1\neq m_2$.
%\end{lemma}
%
%
%\begin{proof}
%
%
%Suppose there is a continuous path $\cT^{(\gamma)}$ connecting $\cT^{(0)}\in\mathfrak{I}_{(r,m_1)}^{(TP)}$ with $\cT^{(1)}\in\mathfrak{I}_{(r,m_2)}^{(TP)}$, $m_1\neq m_2$. Along such a path the eigenvalues of $\cT^{(\gamma)}$ would change continuously. But any $\cT^{(0)}\in\mathfrak{I}_{(r,m_1)}^{(TP)}$ has peripheral spectrum $\{e^{2\pi i \frac{k}{m_1}}\}_{k=1,...,m_1}$ and any $\cT^{(1)}\in\mathfrak{I}_{(r,m_2)}^{(TP)}$
%has peripheral spectrum $\{e^{2\pi i \frac{k}{m_2}}\}_{k=1,...,m_2}$.

%\end{proof}
%
%
%Given the above it is natural to ask if $\mathfrak{I}_{(r,m)}^{(TP)}$ are connected components of $\mathfrak{I}_{(r)}^{(TP)}$ and if not what are these. 

\end{document}